\documentclass[conference]{IEEEtran}
\IEEEoverridecommandlockouts

\ifCLASSINFOpdf
  \usepackage[pdftex]{graphicx}
  \usepackage{epstopdf}
\else
  \usepackage[dvips]{graphicx}
  \usepackage{epstopdf}
\fi

\usepackage[hyphens]{url}
\usepackage[cmex10]{amsmath}
\usepackage{amssymb}
\usepackage[usenames,dvipsnames]{color}

\usepackage{setspace}
\usepackage{graphicx}
\usepackage{float}
\usepackage[caption=false,font=footnotesize]{subfig}
\usepackage{wrapfig}
\usepackage{algorithmic}
\usepackage[algoruled, linesnumbered]{algorithm2e} 
\usepackage{mathrsfs}
\usepackage{cite}
\usepackage{array}
\usepackage{mdwmath}
\usepackage{mdwtab}
\usepackage{fixltx2e}
\usepackage{amsthm}
\usepackage[margin=0.7in]{geometry}
\usepackage{tikz}
\usepackage{scalerel}
\allowdisplaybreaks
\newtheorem{theorem}{Theorem}
\newtheorem{lemma}{Lemma}

\newtheorem{remark}{Remark}[section]

\allowdisplaybreaks

\begin{document}
%
\title{Performance of Analog Beamforming Systems with Optimized Phase Noise Compensation}
\author{Vishnu~V.~Ratnam,~\IEEEmembership{Member,~IEEE.}
\thanks{V. V. Ratnam is with the Standards and Mobility Innovation Lab at Samsung Research America, Plano, TX, 75023 USA. (e-mail: vishnu.r@samsung.com, ratnamvishnuvardhan@gmail.com).} 
}
\maketitle

\begin{abstract}
Millimeter-wave and Terahertz frequencies, while promising high throughput and abundant spectrum, are highly susceptible to hardware non-idealities like phase-noise, which degrade the system performance and make transceiver implementation difficult. While several phase-noise compensation techniques have been proposed, there are limited results on the post-compensation system performance. Consequently, in this paper, a generalized reference-signal (RS) aided low-complexity phase-noise compensation technique is proposed for high-frequency, multi-carrier systems. The technique generalizes several existing solutions and involves an RS that is transmitted in each symbol, occupies adjacent sub-carriers, and is separated from the data by null sub-carriers. A detailed theoretical analysis of the post-phase-noise compensation performance is presented for an analog beamforming receiver under an arbitrary phase-noise model. Using this analysis, the performance-impact of several system parameters is examined and the throughput-optimal designs for the RS sequence, RS bandwidth, power allocation, number of null sub-carriers, and the number of estimated phase-noise spectral components are also derived. Simulations performed under 3GPP compliant settings suggest that the proposed scheme is robust to phase-noise modeling errors and can, with the optimized parameters, achieve better performance than several existing solutions.
\end{abstract}

\IEEEpeerreviewmaketitle

\section{Introduction} \label{sec_intro}
Millimeter (mm) and Terahertz (THz) frequencies offer a huge increase in bandwidth in comparison to the sub-$6$GHz frequencies, and are thus strong candidate communication bands to successfully deliver the exponentially rising wireless data traffic \cite{Cisco_VNI}. These higher frequencies also enable implementation of massive antenna arrays on small form factors, making massive multiple-input-multiple-output (MIMO) systems practically more viable. However, hardware non-idealities like phase-noise (PhN) also tend to increase with the carrier frequency, degrading the system performance. The PhN arises from the device noise in the radio-frequency oscillator circuit of the transceiver, and causes random perturbations in the instantaneous frequency of the oscillator output \cite{Galton2019}. While low levels of PhN can cause a slow channel aging effect, also known as common phase error (CPE) \cite{Krishnan2014, Papazafeiropoulos2017}, higher levels of PhN additionally induce symbol distortion. This distortion manifests as inter-carrier-interference (ICI) in multi-carrier systems, such as those operating with orthogonal frequency division multiplexing (OFDM) \cite{Armada2001, Piazzo2002, Wu2002, Wu2006}, and can severely limit the signal-to-interference-plus-noise ratio (SINR) gains offered by the massive antenna arrays, at mm-wave/THz frequencies.\footnote{With 3GPP adopting OFDM at mm-wave frequencies \cite{3gpp_ts39211}, it is expected that multi-carrier systems such as OFDM will also be prime candidates for adoption at THz frequencies.} Consequently, the impact of PhN on multi-carrier massive antenna transceivers has received significant attention in the recent works. 

Information theoretic bounds on the capacity of channels affected by PhN have been explored for a few scenarios and PhN models in \cite{Lapidoth2002, Ghozlan2017}. The performance degradation due to PhN in different multi-antenna architectures has also been investigated in \cite{Pitarokoilis2012, Krishnan2014, Puglielli2016, Chen2017, Corvaja2017, Ratnam_PACE, Ratnam_Globecom2018, Ratnam_CACE}. 
To alleviate the degradation, several PhN estimation and compensation techniques for have also been proposed \cite{Wu2002, Wu2006, Bittner2007, Petrovic2007, Zou2007, Mehrpouyan2012, Robertson1995, 3gpp_ts39211, Casas2002, Leshem2017, Corvaja2009, Rabiei2010, Khanzadi2013, Tanany2001, Jansen2007, Randel2010}, as summarized in the next paragraph. However, despite the plethora of PhN estimation techniques, there are few analytical results on the achievable performance after PhN compensation. 
Such an analysis can help not only in determining the impact and optimal values of the parameters for the PhN estimation scheme, but also in determining the best transceiver architecture to adopt \cite{Chen2017}. As a step towards addressing this gap, in this paper we propose a low-complexity PhN estimation technique that generalizes several existing schemes and enjoys some benefits over others. We then present a rigorous analysis of the performance after the PhN compensation, and determine the throughput maximizing values of many system parameters.

Before discussing the proposed PhN estimation scheme, we provide a brief overview of the prior art. In one class of PhN estimation techniques, the PhN and data are estimated iteratively using a decision feedback estimation process \cite{Wu2002, Wu2006, Bittner2007, Petrovic2007, Zou2007, Mehrpouyan2012} for every symbol. This approach focuses on estimating PhN with minimal pilot overhead, at the expense of a higher computational cost and decoding latency \cite{Rabiei2010}. 
In another class of techniques, pilot aided, non-iterative PhN estimation approaches have been suggested \cite{Robertson1995, Casas2002, Leshem2017, Rabiei2010, Corvaja2009, Khanzadi2013, Tanany2001, Jansen2007, Randel2010}.\footnote{Some of these works also involve an optional decision-directed iterative feedback loop for better performance but with higher complexity.} For example, \cite{Robertson1995} explores the mitigation of CPE due to PhN using pilot sub-carriers spread across the system bandwidth, and a similar strategy was also adopted by the 3GPP Rel 15 new radio (NR) standard \cite{3gpp_ts39211}; \cite{Casas2002} explores the use of time domain pilots for PhN estimation and \cite{Leshem2017} additionally identifies the dominant PhN basis to estimate the principal components; \cite{Corvaja2009, Rabiei2010} explore pilot aided joint CPE and channel estimation, with linear interpolation between the pilots; and \cite{Khanzadi2013} explores a pilot aided Bayesian estimation of PhN via Kalman smoothing. In yet another approach \cite{Tanany2001, Jansen2007, Randel2010}, a high power sinusoidal pilot/reference signal (RS), separated from the data sub-carriers by a guard-band, is used to estimate and compensate for both CPE and ICI in the frequency domain. Recently, the use of such a sinusoidal RS for joint analog channel estimation and PhN compensation in analog beamforming systems has also been explored \cite{Ratnam_Globecom2018, Ratnam_CACE}. 
While most of these works use simulations to quantify the post-PhN-compensation performance \cite{Armada2001, Petrovic2007, Mehrpouyan2012, Wu2002, Krishnan2014, Wu2006, Bittner2007, Leshem2017, Corvaja2009, Khanzadi2013, Tanany2001, Jansen2007, Randel2010, Guo2017, Qi2018}, some analyze the CPE correction case \cite{Wu2004, Stefanatos2017} and some others use approximations involving first-order Taylor approximations to the PhN or assume independence of PhN estimation errors and channel noise \cite{Rabiei2010, Zou2007}. 

\begin{figure}[!htb]
\centering
\includegraphics[width= 0.45\textwidth]{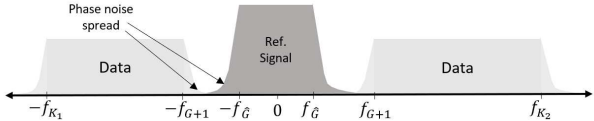}
\caption{An illustration of the base-band, pre-beamforming TX signal with phase noise.}
\label{Fig_TX_sig_illustate}
\end{figure}
In our proposed PhN estimation technique, the transmitter transmits a known band-limited pilot/RS in every symbol along with the data. The RS occupies a few adjacent sub-carriers and is separated from the data sub-carriers by null sub-carriers, as shown in Fig. \ref{Fig_TX_sig_illustate}. The receiver uses the received RS to estimate the dominant spectral components of PhN in the frequency domain. The benefits of the proposed scheme are as follows: firstly, unlike the decision feedback iterative approaches \cite{Wu2002, Wu2006, Bittner2007, Petrovic2007, Zou2007, Mehrpouyan2012}, the use of pilots reduces the computational complexity of PhN estimation and prevents error propagation. Secondly, unlike the pilots in \cite{Robertson1995, Casas2002, Leshem2017, Rabiei2010, Corvaja2009, Khanzadi2013}, the band-limitedness of the proposed RS obviates the need for channel equalization prior to PhN estimation, thus reducing estimation noise accumulation. Thirdly, the use of null sub-carriers and frequency domain PhN estimation prevent ICI from the data sub-carriers and inter-symbol interference (ISI) from adjacent symbols, respectively, during PhN estimation. Finally, the use of multiple sub-carriers for the RS spreads out the RS power spectrum, which can help meet spectral mask regulations with ease in comparison to \cite{Tanany2001, Jansen2007, Randel2010}, which are special cases of the proposed estimation technique with a sinusoidal RS. 
In this paper, we consider a multi-antenna system with receive analog beamforming and an arbitrary oscillator PhN model, and derive closed-form bounds for the SINR and throughput after PhN compensation with this technique. The bounds are used to analyze the impact of different system parameters on the performance. 
Furthermore, we also find the throughput-optimal values\footnote{Throughout the paper, we use `throughput-optimal' to refer to solutions that maximize an approximate bound to the system capacity after accounting for the RS overhead.} of some of these parameters, namely, the power allocation to the RS, the number of RS sub-carriers and null sub-carriers to use and the number of PhN spectral components to estimate. We also show that sequences with good aperiodic auto-correlation properties, e.g. Barker sequences \cite{Golomb1965, Leukhin2018} and aperiodic Zero Correlation Zone (ZCZ) sequences \cite{Donelan2002, Han2007, Wu2005}, are good candidates for the RS. 
The proposed technique is also consistent with the PhN estimation framework in 3GPP NR Rel 15, where such an RS is referred to as the \emph{phase-tracking reference signal} (PTRS) \cite{Guo2017, Qi2018, 3gpp_ts39211}. However, the current 3GPP PTRS has a different time-frequency pattern and only estimates and compensates for the CPE. 
The contributions of this paper are as follows:
\begin{enumerate}
\item We propose a generalized RS aided PhN compensation technique for multi-antenna OFDM systems.
\item We characterize the SINR and throughput of a multi-antenna receiver using analog beamforming and the PhN compensation technique, under an arbitrary PhN model and in terms of easily computable PhN statistics.
\item These PhN statistics are also derived in closed form for an important PhN model. 
\item We find throughput-optimal designs for the RS and also provide throughput-optimal solutions for the RS bandwidth, power allocation, the number of null sub-carriers and the number of PhN spectral components to estimate with the technique.
\item We compare of the performance of the scheme to several prior works under practically relevant simulation scenarios.
\end{enumerate}
The organization of the paper is as follows: the system model is discussed in Section \ref{sec_chan_model}; the PhN models and PhN spectral statistics are discussed in \ref{sec_PhN_stats}; the PhN estimation and compensation is discussed in Section \ref{sec_PhN_est_and_comp}; the signal, interference and noise components of the demodulated outputs are characterized in Section \ref{sec_anal_output}; the system throughput is studied in Section \ref{sec_perf_anal}; the optimal RS design is presented in Section \ref{sec_RS_design}; simulation results are provided in Section \ref{sec_sim_results}; and the conclusions are summarized in Section \ref{sec_conclusions}.

\textbf{Notation:} scalars are represented by light-case letters; vectors and matrices by bold-case letters; and sets by calligraphic letters. Additionally, ${\rm j} = \sqrt{-1}$, $\mathbb{E}\{\}$ represents the expectation operator, $c^*$ is the complex conjugate of a complex scalar $c$, $\|\mathbf{a}\|$ represents the $\ell$-2 norm of a vector $\mathbf{a}$, ${\mathbf{A}}^{\dag}$ is the Hermitian transpose of a complex matrix $\mathbf{A}$, $\lambda^{\downarrow}_a \{\mathbf{A}\}$ is the $a$-th largest eigenvalue of a matrix $\mathbf{A}$, $\delta(t)$ represents the Dirac delta function, $\delta^{A}_{a,b}$ is the modulo-$A$ Kronecker delta function with $\delta^{A}_{a,b} = 1$ if $a=b \ ({\rm mod} \ A)$ and $\delta^{A}_{a,b} = 0$ otherwise and $\mathrm{Re}\{\cdot\}$/$\mathrm{Im}\{\cdot\}$ refer to the real/imaginary components, respectively. 

\section{General Assumptions and System model} \label{sec_chan_model}
\begin{figure}[h] 
\centering 
\includegraphics[width= 0.45\textwidth]{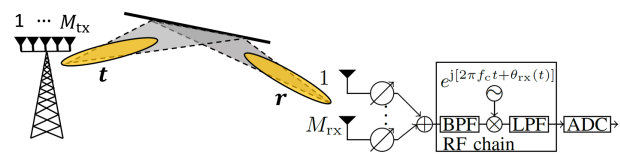}
\caption{An illustration of a single user downlink system with RX analog beamforming.}
\label{Fig_subarray_hybrid_arch}
\end{figure}
We consider the downlink of a single-user system, with one base-station/transmitter (TX) having $M_{\rm tx}$ antennas and one user-equipment/receiver (RX). The RX is assumed to have an analog beamforming architecture \cite{Molisch_HP_mag, Park2017a}, with ${M}_{\rm rx}$ antennas connected to one down-conversion chain via ${M}_{\rm rx}$ phase-shifters, as illustrated in Fig.~\ref{Fig_subarray_hybrid_arch}. While the TX may have an arbitrary architecture, we assume that the TX allocates one up-conversion chain and one data-stream to this RX. 
The TX and RX are assumed to create narrow beams using an arbitrary beamforming technique \cite{Molisch_VarPhaseShift, Alkhateeb2015, Sohrabi2016, Sudarshan2006, Vishnu_ICC2017, Caire2017, Ratnam_HBwS_jrnl, Park2017, Arora2020, Tsinos2020}, and are assumed to have sufficient prior channel knowledge to implement the analog beamformers.\footnote{Such channel knowledge may include either instantaneous channel parameters or average channel parameters \cite{Molisch_HP_mag}.} 
The results in the paper can also directly be extended to a multi-user scenario, where the TX provides orthogonal access to the different users in time, frequency or space. 
For this RX, the TX transmits OFDM symbols with $K$ sub-carriers. The sub-carriers are indexed as $\mathcal{K} \triangleq \{-K_1,...,K_2\}$ for convenience, where $K_1,K_2$ can be arbitrary and $K = K_1+K_2+1$. Data is transmitted on the $K_1-G$ lower and $K_2-G$ higher sub-carriers, i.e. on indices $\mathcal{K} \setminus \mathcal{G}$, where $\mathcal{G} \triangleq \{-G,..,0,..,G\}$ and $G$ is a system parameter determining the number of non-data sub-carriers. The RS for tracking PhN is transmitted on sub-carriers $\hat{\mathcal{G}} \triangleq \{-\hat{G},..,0,..,\hat{G}\}$, where $\hat{G} \leq G$, while the remaining sub-carriers $\mathcal{G} \setminus \hat{\mathcal{G}}$ are nulled to act as a guard-band between the reference and data signals. Here $\hat{G}, G$ are system parameters to be optimized later and the signal structure is illustrated in Fig.~\ref{Fig_TX_sig_illustate} for convenience. 
Under these conditions, the complex equivalent transmit signal of the $0$-th OFDM symbol to the representative RX can be expressed as:
\begin{eqnarray} \label{eqn_tx_signal}
\mathbf{s}_{\rm tx}(t) &=& \sqrt{\frac{2}{T_{\rm s}}}\mathbf{t} \bigg[ \sum_{g \in \hat{\mathcal{G}}} p_{g} e^{{\rm j} 2 \pi f_{g} t} \nonumber \\
&& + \sum_{k \in \mathcal{K} \setminus \mathcal{G}} \!\! x_k e^{{\rm j} 2 \pi f_{k} t} \bigg] e^{{\rm j} [2 \pi f_{\rm c} t + \theta_{\rm tx}(t)]},
\end{eqnarray} 
for $-T_{\rm cp} \leq t \leq T_{\rm s}$, where $T_{\rm s}$ and $T_{\rm cp}$ are the symbol duration and the cyclic prefix duration, respectively, $\mathbf{t}$ is the $M_{\rm tx}\times 1$ unit-norm TX beamforming vector, $p_g$ is the coefficient of the RS on the $g$-th subcarrier, $x_k$ is the data signal on the $k$-th sub-carrier, $f_{\rm c}$ is the carrier frequency, $f_k = k/T_{\rm s}$ represents the frequency offset of the $k$-th sub-carrier from $f_{\rm c}$ and $\theta_{\rm tx}(t)$ represents the PhN process of the TX oscillator. Here we define the \emph{complex equivalent} signal such that the actual (real) transmit signal is given by $\mathrm{Re}\{\mathbf{s}_{\rm tx}(t)\}$. For the data sub-carriers ($k \in \mathcal{K}\setminus \mathcal{G}$), we assume the use of independent, zero-mean data streams with equal power allocation $\mathbb{E}\{{|x_k|}^2\} = E_{\rm d}$ and for RS, we define $E_{\rm r} \triangleq \sum_{g \in \hat{\mathcal{G}}} {|p_g|}^2$. 
The transmit power constraint is then given by $E_{\rm r} + (K-|\mathcal{G}|) E_{\rm d} \leq E_{\rm s}$, where $E_{\rm s}$ is the total OFDM symbol energy (excluding the cyclic prefix). In addition, to emulate spectral mask regulations, we also consider a bound $\bar{E}$ on the transmit power per sub-carrier, i.e., ${|p_g|}^2, E_{\rm d} \leq \bar{E}$. With a slight abuse of notation, throughout the paper we shall use $p_k$ with an unrestricted subscript $k$, with the understanding that $p_k = 0$ if $|k| > \hat{G}$. For convenience, we also define the aperiodic auto-correlation function of the RS sequence as: $\mathscr{R}_{p}(a) \triangleq \sum_{g = -\hat{G}}^{\hat{G}} p_g p^{*}_{g+a}$.

The channel to the representative RX is assumed to have $\tilde{L}$ multi-path components (MPCs), and the corresponding $M_{\rm rx} \times M_{\rm tx}$ channel impulse response matrix and its Fourier transform, respectively, are given as \cite{Akdeniz2014}:
\begin{subequations} \label{eqn_channel_impulse_resp}
\begin{eqnarray} 
\mathbf{H}(t) &=& \sum_{\ell=0}^{\tilde{L}-1} \alpha_{\ell} \mathbf{a}_{\rm rx}(\ell) {\mathbf{a}_{\rm tx}(\ell)}^{\dag} \delta(t - \tau_{\ell}) \\
\boldsymbol{\mathcal{H}}(f) &=& \sum_{\ell=0}^{\tilde{L}-1} \alpha_{\ell} \mathbf{a}_{\rm rx}(\ell) {\mathbf{a}_{\rm tx}(\ell)}^{\dag} e^{- {\rm j} 2 \pi (f_{\rm c}+ f) \tau_{\ell}},
\end{eqnarray}
\end{subequations}
where $\alpha_{\ell}$ is the complex amplitude, $\tau_{\ell}$ is the delay and $\mathbf{a}_{\rm tx}(\ell), \mathbf{a}_{\rm rx}(\ell)$ are the $M_{\rm tx} \times 1$ TX and $M_{\rm rx} \times 1$ RX response vectors, respectively, of the $\ell$-th MPC. 
For example, the $\ell$-th RX response vector for a uniform planar antenna array with $M_{\rm H}$ horizontal and $M_{\rm V}$ vertical elements (${M}_{\rm rx} = {M}_{\rm H} {M}_{\rm V}$) is given by $\mathbf{a}_{\rm rx}(\ell) = \tilde{\mathbf{a}}_{\rm rx}\big(\psi^{\rm rx}_{\rm azi}(\ell), \psi^{\rm rx}_{\rm ele}(\ell)\big)$, where: 
\begin{align} \label{eqn_array_response_planar}
{[\tilde{\mathbf{a}}_{\rm rx}\big(\psi_{\rm azi}, \psi_{\rm ele}\big)]}_{{M}_{\rm V}h + v} &= \exp{\bigg\{{\rm j} 2 \pi \frac{\Delta_{\rm H} h \sin[\psi_{\rm azi}]\sin[\psi_{\rm ele}]}{\lambda}} \nonumber \\
& + {\rm j} 2 \pi \frac{\Delta_{\rm V} (v-1) \cos[\psi_{\rm ele}]}{\lambda} \bigg\}, 
\end{align}
for $h \in \{0,..,{M}_{\rm H}-1\}$ and $v \in \{1,..,{M}_{\rm V}\}$, $\psi^{\rm rx}_{\rm azi}(\ell)$, $\psi^{\rm rx}_{\rm ele}(\ell)$ are the azimuth and elevation angles of arrival for the $\ell$-th MPC, $\Delta_{\rm H}, \Delta_{\rm V}$ are the horizontal and vertical antenna spacings and $\lambda$ is the carrier wavelength. The expression for $\mathbf{a}_{\rm tx}(\ell)$ can be obtained similarly. 
Without loss of generality, let the channel MPCs with non-negligible power along the TX-RX analog beams be indexed as $\{0,..,L-1\}$, where $L \leq \tilde{L}$. Due to the large antenna arrays and associated narrow analog beams at the TX and RX, the effective channel $\mathbf{r}^{\dag} \mathbf{H}(t) \mathbf{t}$ typically has a small delay spread (i.e., $\tau_{L-1} \ll T_{\rm s}$) and a large coherence bandwidth \cite{Wyne2011}. Consequently, and since the TX can afford an accurate oscillator, we shall neglect variation of the TX PhN within this small delay spread, i.e., $\theta_{\rm tx}(t-\tau_{L-1}) \approx \theta_{\rm tx}(t)$. Additionally, we shall also assume that the non-data subcarriers ($k \in \mathcal{G}$) lie within a coherence bandwidth of the effective channel i.e. $T_{\rm s}/2g \geq \tau_{L-1}$.

The RX is assumed to have a low noise amplifier followed by a band-pass filter (BPF) at each antenna, that leaves the desired signal un-distorted but suppresses the out-of-band noise. The filtered signals at each antenna are then phase-shifted by an RX analog beamformer, combined and down-converted via an RX oscillator and then sampled by an analog-to-digital converter (ADC) at $K/T_{\rm s}$ samples/sec, as depicted in Fig.~\ref{Fig_subarray_hybrid_arch}. Assuming perfect timing synchronization at the RX, the sampled received base-band signal for the $0$-th OFDM symbol can be expressed as: 
\begin{eqnarray} \label{eqn_rx_signal}
s_{\rm rx, BB}[n] = \sum_{\ell=0}^{L-1} \alpha_{\ell} {\mathbf{r}}^{\dag} \mathbf{a}_{\rm rx}(\ell) {\mathbf{a}_{\rm tx}(\ell)}^{\dag} \mathbf{t} \bigg[ \sum_{k \in \hat{\mathcal{G}}} p_{k} e^{{\rm j} 2 \pi f_{k} (\frac{nT_{\rm s}}{K} -\tau_{\ell})}  \!\!\!\! \!\!\!\! \nonumber \\
 + \!\!\!\sum_{k \in \mathcal{K} \setminus \mathcal{G}} \!\!\! x_k e^{{\rm j} 2 \pi f_{k} (\frac{nT_{\rm s}}{K}-\tau_{\ell})} \bigg] e^{{\rm j} [\theta_{\rm tx}[n] + \theta_{\rm rx}[n]]} + w[n], \!\!\!\!
\end{eqnarray}
for $0 \leq n < K$, where the $\mathrm{Re}/\mathrm{Im}$ parts of $s_{\rm rx, BB}[n]$ are the outputs corresponding to the in-phase and quadrature-phase components of the RX oscillator, $\mathbf{r}$ is the ${M}_{\rm rx} \times 1$ unit norm RX beamformer, $\theta_{\rm tx}[n] \triangleq \theta_{\rm tx}(nT_{\rm s}/K)$, $\theta_{\rm rx}[n]$ is the sampled PhN process of the RX oscillator, $w[n] \sim \mathcal{CN}(0, \mathrm{N}_0 K)$ is the post-beamforming \emph{effective} additive Gaussian noise process with independent and identically distributed samples and $\mathrm{N}_0$ is the noise power spectral density. Conventional OFDM demodulation is then performed on the combined base-band signal \eqref{eqn_rx_signal}. The demodulated sub-carriers $k \in \mathcal{G}$ are used for PhN estimation, while sub-carriers $k \in \mathcal{K} \setminus \mathcal{G}$ are used for data demodulation, as shall be discussed in Section \ref{sec_PhN_est_and_comp}. Prior estimates of the effective channel $\mathbf{r}^{\dag} \mathbf{H}(t) \mathbf{t}$ are not required for the proposed PhN estimation algorithm in Section \ref{subsec_PhN_est}, which prevents accumulation of estimation noise. 

\subsection{Phase Noise Model}
The PhN of a free-running oscillator is often modeled as a Wiener process \cite{Piazzo2002, Wu2006, Petrovic2007}. In practice, however, oscillators are usually driven by a phase lock loop (PLL) to reduce the output PhN. Several models have been proposed for the PhN of such PLL circuits: from the Ornstein-Ulhenbeck (OU) process \cite{Petrovic2007, Mehrotra2002} used in theoretical investigations, to the filtered Gaussian models used in system level simulations \cite{TR38803_PhnModel} etc. To keep the analysis compatible with any such PhN models, in this work we model $\theta[n] \triangleq \theta_{\rm tx}[n] + \theta_{\rm rx}[n]$ as an arbitrary random process satisfying the following criteria:
\begin{enumerate}
\item[C1] $e^{-{\rm j} \theta[n]}$ is a wide-sense stationary process.
\item[C2] The power spectrum of $e^{-{\rm j}\theta[n]}$ is dominated by the low frequency components.
\end{enumerate} 
All of the aforementioned PhN models satisfy these two generic criteria. Finally, since oscillator data sheets usually report the PhN using metrics that quantify its power spectral density \cite{mmMagic}, we shall strive to present our results only in terms of the power spectral density of $e^{-{\rm j}\theta[n]}$. 

\section{Phase noise and channel noise statistics} \label{sec_PhN_stats}
In this section, we analyze the statistics of the PhN and the channel noise. Note that the sampled channel noise $w[n]$ and the sampled PhN $e^{{\rm j} \theta[n]}$ for $0 \leq n < K$ can be expressed using their normalized Discrete Fourier Transform (nDFT) coefficients as:
\begin{subequations} \label{eqn_DFT_coeffs}
\begin{eqnarray}
w[n] &=& \sum_{k \in \mathcal{K}} W_k e^{{\rm j} 2 \pi k n/K} \label{eqn_DFT_coeffcs_noise} \\
e^{{\rm j} \theta[n]} &=& \sum_{k \in \mathcal{K}} \Omega_k e^{{\rm j} 2 \pi k n/K}, \label{eqn_DFT_coeffcs_PhN}
\end{eqnarray}
\end{subequations}
where $W_{k} = \frac{1}{K} \sum_{n=0}^{K-1} w[n] e^{-{\rm j} 2 \pi k n/K}$ and $\Omega_{k} = \frac{1}{K} \sum_{n=0}^{K-1} e^{-{\rm j} \theta[n]} e^{-{\rm j} 2 \pi k n/K}$ are the corresponding nDFT coefficients. Here nDFT is a slightly unorthodox definition for Discrete Fourier Transform, where the normalization by $K$ is performed while finding $W_{k}, \Omega_{k}$ instead of in \eqref{eqn_DFT_coeffs}. These coefficients satisfy the following remark:
\begin{remark} \label{Lemma_PhN_properties}
The nDFT coefficients $W_{k}, \Omega_{k}$ are periodic with period $K$ and satisfy:
\begin{subequations}
\begin{eqnarray}
\sum_{k \in \mathcal{K}} \Omega_{k_1+k} \Omega_{k_2+k}^{*} &=& \delta^{K}_{k_1,k_2} \label{eqn_lemma_1} \\
\mathbb{E}\{ W_{k_1} {W_{k_2}}^{\dag}\} &=& \delta_{k_1,k_2}^{K} \mathrm{N}_0, \label{eqn_lemma_N_properties}
\end{eqnarray}
\end{subequations}
for arbitrary integers $k_1,k_2$.
\end{remark}
The proof is skipped for brevity, and can be found in \cite{Ratnam_CACE}. 
Note that \eqref{eqn_lemma_1} implies that the nDFT coefficients are coupled together, which makes the analysis difficult -- a fact neglected in first-order Taylor approximation based analysis \cite{Rabiei2010}. 
Using the fact that $e^{-{\rm j} \theta[n]}$ is wide-sense stationary from criterion C1, we also define the second order statistical parameters: $\Delta_{k_1,k_2} \triangleq \mathbb{E}\{\Omega_{k_1} \Omega_{k_2}^{*}\} $ and $\mu(a, b) \triangleq \sum_{c = a - b}^{a+b} \Delta_{c,c}$. Note that from criterion C2, we have $\Delta_{k,k} \approx 0$ and $\mu(0,a) \approx 1$ for $0 \ll |k|, |a| \leq K/2$. These statistics can be computed for an arbitrary PhN process using Monte Carlo simulations. However for some important cases, such as the Wiener PhN model, they can be computed in closed from as shown in Appendix \ref{appdix_wiener}. The analysis and observations from Appendix \ref{appdix_wiener} also justify criterion C2.

\section{Phase noise estimation and compensation} \label{sec_PhN_est_and_comp}
In this section, we discuss the PhN estimation and compensation approach. From the definition of nDFT coefficients and from \eqref{eqn_rx_signal}, note that the received signal on sub-carrier $k$ can be expressed as:
\begin{eqnarray} 
Y_{k} &=& \frac{1}{K}\sum_{n=0}^{K-1} s_{\rm rx, BB}[n] e^{- {\rm j} 2 \pi k n/K} \nonumber \\
& \approx & \sum_{g \in \hat{\mathcal{G}}} \beta_{0} p_{g} \Omega_{k - g} + \!\!\! \sum_{\bar{k} \in \mathcal{K} \setminus \mathcal{G}} \!\!\! \beta_{\bar{k}} x_{\bar{k}} \Omega_{k - \bar{k}} + W_k, \label{eqn_Yk}
\end{eqnarray}
where $\beta_k \triangleq \mathbf{r}^{\dag} \boldsymbol{\mathcal{H}}(f_k) \mathbf{t}$ and we use $\beta_{g} \approx \beta_{0}$ for $g \in \hat{\mathcal{G}}$ as discussed in Section \ref{sec_chan_model}. 
As is evident from \eqref{eqn_Yk}, the transmit signal $x_{\bar{k}}$ on sub-carrier $\bar{k}$ leaks into the received signal of a neighboring sub-carrier $k$. This leakage is in proportion to the nDFT coefficient $\Omega_{k - g}$ and causes ICI. Such ICI can be suppressed by appropriate PhN estimation and compensation, as shall be discussed in the following subsections.

\subsection{Phase noise estimation} \label{subsec_PhN_est}
From criterion C2, note that the PhN nDFT coefficients $\Omega_k$ for the lower frequency indices dominate its behavior and impact. Consequently, we shall only estimate $\Omega_{k}$ for the dominant spectral components: $k \in \mathcal{U}$, where we define $\mathcal{U} \triangleq \{-U,..,U\}$. Here $U$ is a design parameter whose throughput-optimal value shall be discussed later in Section \ref{sec_RS_design}. These coefficients $\{\Omega_{k} \vert k \in \mathcal{U} \}$ shall be estimated from the received RS signal, via the sub-carrier outputs $\{Y_k \mathrel{\stretchto{\mid}{3ex}} |k| \leq \hat{G}+U\}$. Note that to suppress the ICI from the data sub-carriers during this estimation process, the number of null sub-carriers in $\mathcal{G}$ can be chosen to be sufficiently large (see Fig.~\ref{Fig_TX_sig_illustate}). Thus, neglecting the interference from these data sub-carriers, from \eqref{eqn_Yk}, the received signal on a sub-carrier $k \in [-\hat{G}-U, \hat{G}+U]$ can be expressed as:
\begin{eqnarray} 
Y_{k} & \stackrel{(1)}{\approx} & \beta_0 \Big[ \sum_{g = -\hat{G}}^{\hat{G}} p_{g} \Omega_{k-g} \Big] + W_{k} \nonumber \\
& = & \beta_0 \Big[ \sum_{u \in \mathcal{U}} p_{k-u} \Omega_{u} \Big] + W_{k} \nonumber \\
&& + \beta_0 \Big[ \sum_{g = -2\hat{G}-U}^{-U-1} p_{k-g} \Omega_{g} + \sum_{g = U+1}^{2\hat{G}+U} p_{k-g} \Omega_{g} \Big],
\label{eqn_Yk_ref}
\end{eqnarray}
where $p_{a} = 0$ if $|a| > \hat{G}$. The first term in \eqref{eqn_Yk_ref} involves PhN components we desire to estimate, the second term is the channel noise, while the last term involves interference from the higher frequency PhN terms. Note that the neglected data interference in ${\scriptstyle\stackrel{(1)}{\approx}}$ is proportional to $E_{\rm d} [1 - \mu(0, G-\hat{G}-U)]$, which rapidly reduces to $0$ with increasing $G$ from criterion C2. Thus, to keep this approximation tight, we assume $G \geq \hat{G} + 2U$ for the rest of the analysis. Equation \eqref{eqn_Yk_ref} can be expressed in matrix form as:
\begin{eqnarray} \label{eqn_PhN_est_mtrx_form}
\mathbf{Y}^{(U)} = \beta_0 \mathbf{P}^{(U)} \boldsymbol{\Omega}^{(U)} + \beta_0 \mathbf{Q}^{(U)} \boldsymbol{\Omega}^{(\rm inf)} + \mathbf{W}^{(U)},
\end{eqnarray}
where $\mathbf{Y}^{(U)}, \mathbf{W}^{(U)}$ are $(2\hat{G}+2U+1) \times 1$ vectors with the $i$-th entries being $Y_{i-\hat{G}-U-1}$ and $W_{i-\hat{G}-U-1}$, respectively, $\boldsymbol{\Omega}^{(U)}$ is a $(2U+1) \times 1$ vector with ${[\boldsymbol{\Omega}^{(U)}]}_{i} = \Omega_{i-U-1}$ and $\mathbf{P}^{(U)}$ is a $(2\hat{G}+2U+1) \times (2U+1)$ rectangular, banded Toeplitz matrix with ${[\mathbf{P}^{(U)}]}_{i,j} = p_{-\hat{G}+i-j}$. 
Furthermore, $\boldsymbol{\Omega}^{(\rm inf)}$ is a $4\hat{G} \times 1$ vector with ${[\boldsymbol{\Omega}^{(\rm inf)}]}_{i} = \Omega_{-2\hat{G}-2U-1+i}$ if $i \leq 2 \hat{G}$ and ${[\boldsymbol{\Omega}^{(\rm inf)}]}_{i} = \Omega_{4\hat{G}+2U+1-i}$ otherwise, and $\mathbf{Q}^{(U)}$ is a $(2\hat{G}+2U+1) \times 4 \hat{G}$ rectangular matrix with ${[\mathbf{Q}^{(U)}]}_{i,j} = p_{\hat{G}+i-j}$ for $j \leq 2 \hat{G}$ and ${[\mathbf{Q}^{(U)}]}_{i,j} = p_{\hat{G}-2U-1+i-j}$ otherwise. For convenience of the reader, the cumbersome \eqref{eqn_PhN_est_mtrx_form} is represented pictorially in Fig.~\ref{Fig_PhN_est_mtrx_form}.
\begin{figure}[!htb]
\centering
\vspace{-0.4cm}
\includegraphics[width= 0.49\textwidth]{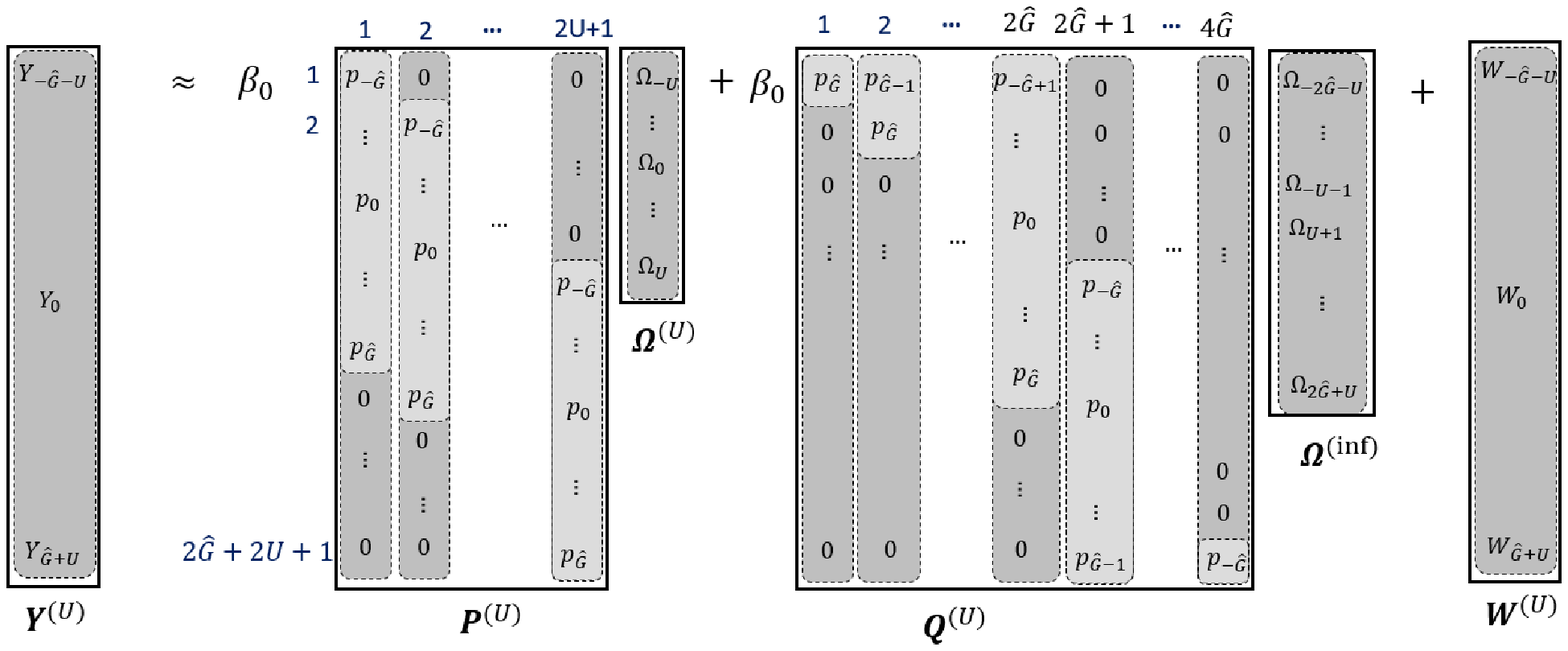} 
\vspace{-0.3cm}
\caption{Pictorial depiction of equation \eqref{eqn_PhN_est_mtrx_form}.}
\label{Fig_PhN_est_mtrx_form}
\end{figure}
Assuming $\mathbf{R}_{p} \triangleq {[\mathbf{P}^{(U)}]}^{\dag} \mathbf{P}^{(U)}$ to be full rank without loss of generality, the least squares (LS) estimate for $\beta_0 \boldsymbol{\Omega}^{(U)}$ can then be obtained as:
\begin{subequations}
\begin{eqnarray} \label{eqn_PhN_est_LS}
\widehat{\beta_{0} \boldsymbol{\Omega}^{(U)}} &=& \mathbf{R}_{p}^{-1} {[\mathbf{P}^{(U)}]}^{\dag} \mathbf{Y}^{(U)} \nonumber \\
&=& \beta_0 (1 + \chi) \boldsymbol{\Omega}^{(U)} + \beta_0 \boldsymbol{\Phi}^{(U)} + \hat{\mathbf{W}}^{(U)},
\end{eqnarray}
where $\hat{\mathbf{W}}^{(U)} \sim \mathcal{CN} \big(\mathbb{O}_{(2U+1) \times 1}, \mathbf{R}_{p}^{-1} \mathrm{N}_0 \big)$ and we define: 
\begin{align}
\chi & \triangleq \mathbb{E} \left\{ {\boldsymbol{\Omega}^{(U)}}^{\dag} \mathbf{R}_{p}^{-1} {[\mathbf{P}^{(U)}]}^{\dag} \mathbf{Q}^{(U)} \boldsymbol{\Omega}^{(\rm inf)} \right\} \Big/ \mu(0,U) \label{eqn_chi_defn} \\ 
\boldsymbol{\Phi}^{(U)} & \triangleq \mathbf{R}_{p}^{-1} {[\mathbf{P}^{(U)}]}^{\dag} \mathbf{Q}^{(U)} \boldsymbol{\Omega}^{(\rm inf)} - \chi \boldsymbol{\Omega}^{(U)}. 
\end{align}
\end{subequations}
Here, $\hat{\mathbf{W}}^{(U)}$ is the estimation noise, $\chi$ quantifies the fraction of interference that is aligned with the desired signal $\boldsymbol{\Omega}^{(U)}$, and $\boldsymbol{\Phi}^{(U)}$ is the uncorrelated component of the interference. While in most practical settings we may have $\chi \approx 0$, here we shall not make that assumption for generality.

While the linear minimum mean square error (LMMSE) estimation of $\boldsymbol{\Omega}^{(U)}$ may lead to less estimation error than with LS estimation, here we consider the latter due to two reasons. Firstly, the LS estimate in \eqref{eqn_PhN_est_LS} does not require knowledge of $\beta_0$, thus preventing channel estimation errors from affecting the PhN estimates, i.e., error propagation. Secondly, unlike LMMSE, the LS estimate does not require the knowledge of the PhN cross-statistics $\{\Delta_{u_1,u_2}| u_1 \neq u_2\}$, which may be unavailable in practice for a realistic PhN process. A comparison of the two estimators is performed via simulations later in Section \ref{sec_sim_results}. For convenience, let us also define $\Phi_u \triangleq {[\boldsymbol{\Phi}^{(U)}]}_{U+1+u}$ and $\hat{W}_{u} \triangleq {[\hat{\mathbf{W}}^{(U)}]}_{u+U+1}$ as the uncorrelated interference and the estimation noise, respectively, in the LS estimate of $\Omega_u$ for $u \in \mathcal{U}$. We then have the following result on $\boldsymbol{\Phi}^{(U)}$:
\begin{lemma} \label{Th_PhN_est_interf}
The interference term in \eqref{eqn_PhN_est_LS}, viz. $\boldsymbol{\Phi}^{(U)}$ satisfies: 
\begin{subequations}
\begin{align}
& \sum_{u \in \mathcal{U}} {| (1 \!+\! \chi)\Omega_u + \Phi_u|}^2 \leq 1 + \Upsilon_p \label{eqn_revLemma_1} \\
& \sum_{u \in \mathcal{U}} \mathbb{E} \left\{ {| \Phi_u|}^2 \right\} \leq \Upsilon_p \mu(0, 2\hat{G}+U) \nonumber \\
& \qquad  \qquad \qquad \qquad - \big(\Upsilon_{p} \!+\! {|\chi|}^2 \big) \mu(0, U) \label{eqn_revLemma_2} \\
& {|\chi|}^2 \leq \frac{\Upsilon_p [\mu(0, 2\hat{G}+U) - \mu(0,U)]}{\mu(0, U)}, \label{eqn_revLemma_4} \\
& \sum_{u \in \mathcal{U}} \mathbb{E}  \Big\{ \Phi_u \Omega_u^{*}  \Big\} = 0 \label{eqn_revLemma_3}
\end{align}
\end{subequations}
where $\Upsilon_p \triangleq \lambda_1^{\downarrow} \big\{ {[\mathbf{Q}^{(U)}]}^{\dag} \mathbf{P}^{(U)} \mathbf{R}_{p}^{-2} {[\mathbf{P}^{(U)}]}^{\dag} \mathbf{Q}^{(U)} \big\}$.\end{lemma}
\begin{proof}
For \eqref{eqn_revLemma_1}, we use:
\begin{flalign}
& \sum_{u \in \mathcal{U}} {| (1 + \chi) \Omega_u + \Phi_u|}^2 = \sum_{u \in \mathcal{U}} {| \Omega_u + (\Phi_u + \chi \Omega_u)|}^2 & \nonumber \\
& \qquad \stackrel{(1)}{\leq}  {\Bigg[ \sqrt{\sum_{u \in \mathcal{U}} {|\Omega_u|}^2} + \sqrt{\sum_{u \in \mathcal{U}} {|\Phi_u + \chi \Omega_u |}^2 } \Bigg]}^2 & \nonumber \\
& \qquad \stackrel{(2)}{=} {\Big[ {\big\| \boldsymbol{\Omega}^{(U)} \big\|} + { \big\| \mathbf{R}_{p}^{-1} {\mathbf{P}^{(U)}}^{\dag} \mathbf{Q}^{(U)} \boldsymbol{\Omega}^{(\rm inf)} \big\| } \Big]}^2 & \nonumber \\
& \qquad \stackrel{(3)}{\leq} {\Bigg[\sqrt{\sum_{u \in \mathcal{U}} {|\Omega_u|}^2 } + \sqrt{ \Upsilon_p {\big\| \boldsymbol{\Omega}^{(\rm inf)} \big\|}^2} \Bigg]}^2 & \nonumber \\
& \qquad \stackrel{(4)}{\leq} {\Bigg[\sqrt{\sum_{u \in \mathcal{U}} {|\Omega_u|}^2 } + \sqrt{ \Upsilon_p \Big(1 - \sum_{u \in \mathcal{U}} {|\Omega_u|}^2 \Big)  } \Bigg]}^2 & \nonumber \\
& \qquad \stackrel{(5)}{\leq} 1 + \Upsilon_p , \nonumber &
\end{flalign}
where ${\scriptstyle\stackrel{(1)}{\leq}}$ follows from the Triangle inequality, ${\scriptstyle\stackrel{(2)}{=}}$ follows from the definition of $\boldsymbol{\Phi}^{(U)}$, ${\scriptstyle\stackrel{(3)}{\leq}}$ follows by defining $\Upsilon_{p}$ as the largest singular value, ${\scriptstyle\stackrel{(4)}{\leq}}$ follows from \eqref{eqn_lemma_1}, and ${\scriptstyle\stackrel{(5)}{\leq}}$ follows by optimizing $\sum_u {|\Omega_u|}^2$. Similarly for \eqref{eqn_revLemma_2} we have:
\begin{flalign}
& \sum_{u \in \mathcal{U}} \mathbb{E} \left\{ {| \Phi_u|}^2 \right\} \stackrel{(6)}{=} \mathbb{E} \Big\{ {\big\| \mathbf{R}_{p}^{-1} {\mathbf{P}^{(U)}}^{\dag} \mathbf{Q}^{(U)} \boldsymbol{\Omega}^{(\rm inf)} \big\|}^2 & \nonumber \\
& \qquad \qquad \qquad \qquad - {|\chi|}^2 {\big\| \boldsymbol{\Omega}^{(u)} \big\|}^2 \Big\} & \nonumber \\
& \qquad \qquad \qquad \ \stackrel{(7)}{\leq} \mathbb{E} \left\{ \Upsilon_p {\big\| \boldsymbol{\Omega}^{(\rm inf)} \big\| }^2 - {|\chi|}^2 { \big\| \boldsymbol{\Omega}^{(u)} \big\| }^2\right\} & \nonumber \\
& \qquad \qquad \qquad \ \stackrel{(8)}{=} \Upsilon_p \mu(0, 2\hat{G}+U) - \big(\Upsilon_{p} + {|\chi|}^2 \big) \mu(0, U) , \nonumber &
\end{flalign}
where ${\scriptstyle\stackrel{(6)}{=}}$, ${\scriptstyle\stackrel{(7)}{\leq}}$, ${\scriptstyle\stackrel{(8)}{=}}$ follow from the definitions of $\chi$, $\Upsilon_p$ and $\mu(\cdot)$, respectively. Note that \eqref{eqn_revLemma_4} is a direct consequence of \eqref{eqn_revLemma_2} in conjunction with $\sum_{u \in \mathcal{U}} \mathbb{E} \left\{ {| \Phi_u|}^2 \right\} \geq 0$. Finally for \eqref{eqn_revLemma_3} we have:
\begin{align}
\sum_{u \in \mathcal{U}} \mathbb{E}  \Big\{ \Phi_u \Omega_u^{*}  \Big\} = & \mathbb{E} \Big\{ {\boldsymbol{\Omega}^{(U)} }^{\dag} \mathbf{R}_{p}^{-1} {\mathbf{P}^{(U)}}^{\dag} \mathbf{Q}^{(U)} \boldsymbol{\Omega}^{(\rm inf)} \nonumber \\
& - \chi {\boldsymbol{\Omega}^{(U)} }^{\dag} \boldsymbol{\Omega}^{(U)} \Big \} & \nonumber \\
 = & 0, \nonumber
\end{align}
where the last step follows from the definition of $\chi$.
\end{proof}

\subsection{Phase noise compensation}
To compensate for the PhN-induced phase rotation and ICI in \eqref{eqn_Yk}, a simple PhN compensation technique is considered, where the post-compensation $k$-th OFDM output can be expressed as:
\begin{eqnarray} \label{eqn_Yk_with_PhN_comp}
\hat{Y}_k &=& \sum_{u \in \mathcal{U}} {\big[\widehat{\beta_{0} \boldsymbol{\Omega}^{(U)}}\big]}_{u+U+1}^{*} Y_{k+u} \nonumber \\
& \stackrel{(1)}{=}& \sum_{u \in \mathcal{U}} [\beta^{*}_{0} (1+\chi^{*})\Omega_{u}^{*} + \beta^{*}_{0} \Phi_{u}^{*} + \hat{W}^{*}_{u}] \bigg[ W_{k+u} \nonumber \\
&& + \sum_{g \in \hat{\mathcal{G}}} \beta_{0} p_{g} \Omega_{k+u-g} + \!\!\! \sum_{\bar{k} \in \mathcal{K} \setminus \mathcal{G}} \!\! \beta_{\bar{k}} x_{\bar{k}} \Omega_{k+u-\bar{k}} \bigg], 
\end{eqnarray}
and ${\scriptstyle\stackrel{(1)}{=}}$ follows from \eqref{eqn_PhN_est_LS}. These PhN compensated sub-carriers $\{\hat{Y}_k | k \in \mathcal{K} \setminus \mathcal{G}\}$ are then used to demodulate the data signals $x_k$. Note that CPE-only compensation \cite{Robertson1995, Abhayawardhanaa2002} is a special case of \eqref{eqn_Yk_with_PhN_comp}, obtained by picking $U=0$. Using \eqref{eqn_lemma_1}, it can be shown that the above technique can completely cancel the PhN in the absence of estimation noise $\hat{W}_{u}$ and for $U \gg 1$ (albeit at a very high pilot overhead). These demodulated outputs for more general system settings are analyzed in the next section. From \eqref{eqn_PhN_est_LS} and \eqref{eqn_Yk_with_PhN_comp}, it can also be verified that the proposed PhN estimation and compensation technique only requires a small overhead of $\leq (2U+1)K$ complex multiplications per OFDM symbol.

\section{Analysis of the demodulated outputs} \label{sec_anal_output} 
We shall split $\hat{Y}_k$ in \eqref{eqn_Yk_with_PhN_comp} into three components as: $\hat{Y}_k = \hat{S}_k + \hat{I}_k + \hat{Z}_k$. The first component $\hat{S}_k$, referred to as the signal component, involves the terms in \eqref{eqn_Yk_with_PhN_comp} containing $x_k$ and not containing the channel noise, PhN estimation errors or interference. The second component $\hat{I}_k$, referred to as the interference component, involves the terms containing $\{p_{g}, x_{\bar{k}} | g \in \hat{\mathcal{G}}, \bar{k} \in \mathcal{K} \setminus \{k\}\}$ or estimation errors in $\widehat{\beta_{0} \boldsymbol{\Omega}^{(U)}}$ and not containing the channel or estimation noise. The third component $\hat{Z}_k$, referred to as the noise component, contains the remaining terms. These signal, interference and noise components are analyzed in the following subsections. 

\subsection{Signal component analysis} \label{subsec_signal_comp_anal}
From \eqref{eqn_Yk_with_PhN_comp}, the signal component for $k \in \mathcal{K} \setminus \mathcal{G}$ can be expressed as:
\begin{eqnarray} \label{eqn_Sk}
\hat{S}_k &=& \sum_{u \in \mathcal{U}} \beta^{*}_{0} \beta_{k} x_{k} (1 + \chi^{*} ) \mathbb{E} \big\{ {|\Omega_{u}|}^2 \big\}. 
\end{eqnarray}
Note that since the coefficient $\sum_{u \in \mathcal{U}} {|\Omega_{u}|}^2$ can be unknown and random for each symbol, we only consider its statistical mean, viz. $\mu(0,U)$, to contribute to the signal component. Thus the signal component is independent of the instantaneous realization of $\{\Omega_u | u \in \mathcal{U}\}$. As is evident, the phase rotation due the PhN is suppressed by the compensation technique and the magnitude of signal component increases with $U$. Taking an expectation with respect to $x_k$, the energy of the signal component can be expressed as:
\begin{eqnarray} \label{eqn_Sk_stats}
\mathbb{E}\{ {|\hat{S}_k |}^2 \} &=& {|\beta_{0} \beta_{k}|}^2 E_{\rm d} {|1+\chi|}^2 {\mu(0,U)}^2.
\end{eqnarray}

\subsection{Interference component analysis} \label{subsec_interf_comp_anal}
From \eqref{eqn_Yk_with_PhN_comp}, the interference component for $k \in \mathcal{K} \setminus \mathcal{G}$ can be expressed as: $\hat{I}_k = \hat{I}_k^{(1)} + \hat{I}_k^{(2)}$, where:
\begin{subequations} \label{eqn_Ik}
\begin{flalign}
& \hat{I}^{(1)}_k = \sum_{u \in \mathcal{U}} \sum_{g \in \hat{\mathcal{G}}} {|\beta_{0}|}^2 \big[(1+\chi^{*})\Omega_{u}^{*} + \Phi_{u}^{*} \big] p_{g} \Omega_{k+u-g} & \\
& \hat{I}_k^{(2)} = \sum_{u \in \mathcal{U}} \sum_{\bar{k} \in \mathcal{K} \setminus [\mathcal{G} \cup \{k\}]} \!\!\!\!\! \beta_0^{*} \beta_{\bar{k}} x_{\bar{k}} \big[(1+\chi^{*})\Omega_{u}^{*} + \Phi_{u}^{*} \big] \Omega_{k+u-\bar{k}} & \nonumber \\
& \quad \qquad + \beta_0^{*} \beta_{k} x_{k} \Big[ \sum_{u \in \mathcal{U}} \Big((1+\chi^{*})\Omega_u^{*} + \Phi_{u}^{*}\Big) \Omega_{u} & \nonumber \\
& \quad \qquad - (1+\chi^{*})\mu(0,U) \Big]. & \label{eqn_Ik_2}
\end{flalign}
\end{subequations}
Note that the last term $(1+\chi^{*}) \mu(0,U)$ in \eqref{eqn_Ik_2} subtracts out the contribution of the signal component \eqref{eqn_Sk}. Using the independent, zero mean assumption on the sub-carrier data, the first and second moment of $\hat{I}_k$, averaged over the PhN and data can be expressed as: 
\begin{subequations} \label{eqn_Ik_stats}
\begin{align}
\mathbb{E}\{\hat{I}_k \} &= \sum_{{u} \in \mathcal{U}} \sum_{{k} \in \hat{\mathcal{G}}} {|\beta_{0}|}^2 p_{{k}}(1+\chi^{*}) \Delta_{k+{u}-{k}, {u}} \\
\mathbb{E}\{{|\hat{I}_k|}^2\} &= \mathbb{E}\{{|\hat{I}^{(1)}_k|}^2\} + \mathbb{E}\{{|\hat{I}^{(2)}_k|}^2\}. \label{eqn_Ik_stats_2}
\end{align}
The terms in \eqref{eqn_Ik_stats_2} can further be bounded as:
\begin{flalign} 
& \mathbb{E}\{{|\hat{I}^{(1)}_k|}^2\} & \nonumber \\
&= \mathbb{E} {\bigg| \sum_{u \in \mathcal{U}} \Big[(1+\chi^{*})\Omega_{u}^{*} + \Phi_{u}^{*} \Big] \Big[ \sum_{g \in \hat{\mathcal{G}}} p_{g} {|\beta_{0}|}^2 \Omega_{k+u-g} \Big] \bigg|}^2 & \nonumber \\
& \stackrel{(1)}{\leq} \mathbb{E} \bigg\{ \Big[ \sum_{u \in \mathcal{U}}  {\big|(1+\chi)\Omega_{u} + \Phi_{u}\big|}^2 \Big] & \nonumber \\
& \qquad \qquad \times \Big[ \sum_{u \in \mathcal{U}} {\Big| \sum_{g \in \hat{\mathcal{G}}} p_{g} {|\beta_{0}|}^2 \Omega_{k+u-g} \Big|}^2 \Big] \bigg\} & \nonumber \\
& \stackrel{(2)}{\leq} (1 + \Upsilon_p) \mathbb{E}\Big[ \sum_{u \in \mathcal{U}} {\Big| \sum_{g \in \hat{\mathcal{G}}} p_{g} {|\beta_{0}|}^2 \Omega_{k+u-g} \Big|}^2 \Big] & \nonumber \\
& \stackrel{(3)}{\leq} \sum_{u \in \mathcal{U}} (1 + \Upsilon_p) {|\beta_{0}|}^4 E_{\rm r} \Big[ \sum_{g \in \hat{\mathcal{G}}} \Delta_{k+u-g, k+u-g} \Big] & \nonumber \\
& = \sum_{g \in \hat{\mathcal{G}}} (1 + \Upsilon_p) E_{\rm r} {|\beta_{0}|}^4 \mu(k-g, U) & \label{eqn_I1k_stats} \\
& \mathbb{E}\{{|\hat{I}^{(2)}_k|}^2\} & \nonumber \\
& \stackrel{(4)}{\leq} \!\!\!\! \sum_{\bar{k} \in \mathcal{K} \setminus \{k\}} \!\!\!\!\! {|\beta_{0} {\beta}_{\bar{k}}|}^2 E_{\rm d} \mathbb{E} {\bigg| \sum_{u \in \mathcal{U}} \big[ (1+\chi^{*})\Omega_{u}^{*} + \Phi_{u}^{*} \big] \Omega_{k+u-\bar{k}} \bigg|}^2 & \nonumber \\
& \qquad \qquad + {|\beta_{0} {\beta}_{k}|}^2 E_{\rm d} \mathbb{E} \bigg| \sum_{u \in \mathcal{U}} \big[ (1+\chi^{*})\Omega_{u}^{*} + \Phi_{u}^{*} \big] \Omega_{u} & \nonumber \\
& \qquad \qquad \qquad \qquad {- (1+\chi^{*})\mu(0,U) \bigg|}^2 & \nonumber \\
& \stackrel{(5)}{\leq}  \sum_{\bar{u},\ddot{u} \in \mathcal{U}} \!\!\! {|\beta_0 \bar{\beta}|}^2 E_{\rm d} \mathbb{E} \bigg\{ \big[(1+\chi^{*})\Omega_{\bar{u}}^{*} + \Phi_{\bar{u}}^{*} \big] \big[(1+\chi)\Omega_{\ddot{u}}+\Phi_{\ddot{u}} \big]& \nonumber \\
& \quad \times \Big[ \sum_{\bar{k} \in \mathcal{K}} \Omega_{k+\bar{u}-\bar{k}} \Omega_{k+\ddot{u}-\bar{k}}^{*} \Big] \bigg\} - {|\beta_0 \bar{\beta}|}^2 E_{\rm d} {|1+\chi|}^2 {\mu(0,U)}^2 & \nonumber \\
& \stackrel{(6)}{=} \sum_{\bar{u} \in \mathcal{U}} {|\beta_0 \bar{\beta}|}^2 E_{\rm d} \mathbb{E} \big\{ {|(1+\chi)\Omega_{\bar{u}}+\Phi_{\bar{u}}|}^2 \big\} & \nonumber \\
& \qquad \qquad - {|\beta_0 \bar{\beta}|}^2 E_{\rm d} {|1+\chi|}^2 {\mu(0,U)}^2 & \nonumber \\
& \stackrel{(7)}{\leq} {|\beta_0 \bar{\beta}|}^2 E_{\rm d} \Big[ (1 + \chi+\chi^{*})\mu(0,U) - {|1+\chi|}^2 {\mu(0,U)}^2 & \nonumber \\
& \qquad \qquad + \Upsilon_p \big( \mu(0, 2\hat{G}+U) - \mu(0, U) \big) \Big] , \!\!\!\!\!\! & \label{eqn_I2k_stats}
\end{flalign}
\end{subequations}
where ${\scriptstyle \stackrel{(1)}{\leq}}$ follows by using the Cauchy-Schwartz inequality; ${\scriptstyle \stackrel{(2)}{\leq}}$ follows by using \eqref{eqn_revLemma_1}; ${\scriptstyle \stackrel{(3)}{\leq}}$ follows by using the Cauchy-Schwartz inequality again; ${\scriptstyle\stackrel{(4)}{\leq}}$ also follows from the independent, zero mean assumption for sub-carrier data and by including summation over $\bar{k} \in \mathcal{G}$ for the first term; ${\scriptstyle\stackrel{(5)}{\leq}}$ follows by defining $\bar{\beta} \triangleq \max_{k \in \mathcal{K}} |\beta_k|$ and using \eqref{eqn_revLemma_3} for the second term; ${\scriptstyle\stackrel{(6)}{=}}$ follows from \eqref{eqn_lemma_1}; and ${\scriptstyle\stackrel{(7)}{\leq}}$ follows from \eqref{eqn_revLemma_2}. 

\subsection{Noise component analysis} \label{subsec_noise_comp_anal}
The noise component of the received signal on sub-carrier $k \in \mathcal{K} \setminus \mathcal{G}$ can be expressed as: 
\begin{subequations} \label{eqn_Z_k_exansion}
\begin{eqnarray}
\hat{Z}_k &=& \hat{Z}_k^{(1)} + \hat{Z}_k^{(2)} + \hat{Z}_k^{(3)} + \hat{Z}_k^{(4)}, \ \text{where:} \nonumber \\ 
\hat{Z}_k^{(1)} &=& \sum_{u \in \mathcal{U}} \sum_{g \in \hat{\mathcal{G}}} \hat{W}^{*}_{u} \beta_{0} p_{g} \Omega_{k+u-g} \nonumber \\
& \stackrel{(1)}{=} & \sum_{v=-\hat{G}-U}^{\hat{G}+U} \sum_{u \in \mathcal{U}} \hat{W}^{*}_{u} \beta_{0} p_{v+u} \Omega_{k-v} \\
\hat{Z}_k^{(2)} &=& \sum_{u \in \mathcal{U}} \sum_{\bar{k} \in \mathcal{K} \setminus \mathcal{G}} \hat{W}^{*}_{u} \beta_{\bar{k}} x_{\bar{k}} \Omega_{k+u-\bar{k}} \label{eqn_Zk_term2} \\
\hat{Z}_k^{(3)} &=& \sum_{u \in \mathcal{U}} \beta_0^{*} [(1+\chi^{*})\Omega_{u}^{*} + \Phi_u^{*}] W_{k+u} \\
\hat{Z}_k^{(4)} &=& \sum_{u \in \mathcal{U}} \hat{W}^{*}_{u} W_{k+u},
\end{eqnarray}
\end{subequations}
where ${\scriptstyle\stackrel{(1)}{=}}$ is obtained by using change of variables $v = g - u$ and letting $p_a = 0$ for $|a| > \hat{G}$. From Remark \ref{Lemma_PhN_properties} and equations \eqref{eqn_Yk_ref}--\eqref{eqn_PhN_est_LS}, it can be readily verified that $\hat{W}_{u}$ and $W_k$ are circularly symmetric, zero-mean Gaussian and mutually independent for $u \in \mathcal{U}, k \in \mathcal{K} \setminus \mathcal{G}$. Additionally, they are also independent of $\{\Omega_{k} | k \in \mathcal{K} \}$. Therefore the first and second moments of the noise signal, averaged over the PhN, channel noise and data signals, can be expressed as:
\begin{subequations} \label{eqn_Z_k_stats}
\begin{eqnarray}
\mathbb{E}\{\hat{Z}_k\} &=& 0 \\
\mathbb{E}\{ {|\hat{Z}_k|}^2 \} &=& \sum_{i=1}^{4} \mathbb{E}\{ {|\hat{Z}^{(i)}_k|}^2 \}. 
\end{eqnarray}
Using $\diamond$ as short hand for $U+1$ to save space, these individual moments can be expressed as:
\begin{flalign}
& \mathbb{E}\{ {|\hat{Z}^{(1)}_k|}^2 \} \stackrel{(1)}{\leq} \Bigg[ \sum_{v=-\hat{G}-U}^{\hat{G}+U} \!\!\!\!\! \Delta_{k-v, k-v} \Bigg] & \nonumber \\
& \qquad \qquad \qquad \times \Bigg[ \sum_{v=-\hat{G}-U}^{\hat{G}+U} \mathbb{E}{\Bigg|\sum_{u \in \mathcal{U}} \hat{W}^{*}_{\bar{u}} \beta_{0} p_{v+\bar{u}} \Bigg|}^2 \Bigg] & \nonumber \\
& \quad \stackrel{(2)}{=} \mu(k, \hat{G}+U) \bigg[ \sum_{v=-\hat{G}-U}^{\hat{G}+U} \sum_{\bar{u},\ddot{u} \in \mathcal{U}} {|\beta_0|}^2 \mathrm{N}_0 & \nonumber \\
& \qquad \quad \times [\mathbf{R}^{-1}_p]_{\diamond +\ddot{u}, \diamond +\bar{u}} p_{v+\bar{u}} p^{*}_{v+\ddot{u}} \bigg] & \nonumber \\
& \quad \stackrel{(3)}{=} \mu(k, \hat{G}+U) \Bigg[ \sum_{\bar{u},\ddot{u} \in \mathcal{U}} \!\!\! {|\beta_0|}^2 \mathrm{N}_0 [\mathbf{R}^{-1}_p]_{\diamond +\ddot{u}, \diamond +\bar{u}} & \nonumber \\
& \qquad \quad \bigg( \sum_{v = -\hat{G}-U}^{\hat{G} + U} \!\!\!\!\! p_{v+\bar{u}} p^{*}_{v+\ddot{u}} \bigg) \Bigg] & \nonumber \\
& \quad \stackrel{(4)}{=} \mu(k, \hat{G}+U) \Big[ \sum_{\bar{u},\ddot{u} \in \mathcal{U}} {|\beta_0|}^2 \mathrm{N}_0 [\mathbf{R}^{-1}_p]_{\diamond +\ddot{u}, \diamond +\bar{u}} & \nonumber \\
& \qquad \quad [\mathbf{R}_p]_{\diamond +\bar{u}, \diamond +\ddot{u}} \Big] & \nonumber \\
& \quad = \mu(k, \hat{G}+U) |\mathcal{U}| {|\beta_0|}^2 \mathrm{N}_0 & \\
& \mathbb{E}\{ {|\hat{Z}^{(2)}_k|}^2 \} = \sum_{\bar{k} \in \mathcal{K} \setminus \mathcal{G}} {|\beta_{\bar{k}}|}^2 E_{\rm d} \mathbb{E}{\Big| \sum_{u \in \mathcal{U}} \hat{W}^{*}_{u} \Omega_{k+u-\bar{k}} \Big|}^2 & \nonumber \\
& \quad \stackrel{(5)}{\leq} \sum_{\bar{k} \in \mathcal{K}} {|\bar{\beta}|}^2 E_{\rm d} \mathbb{E}{\Big[ \sum_{\bar{u}, \ddot{u} \in \mathcal{U}} \hat{W}^{*}_{\bar{u}} \hat{W}_{\ddot{u}} \Omega_{k+\bar{u}-\bar{k}} \Omega_{k+\ddot{u}-\bar{k}}^{*} \Big]} & \nonumber \\
& \quad = \sum_{\bar{u},\ddot{u} \in \mathcal{U}} {[\mathbf{R}^{-1}_{p}]}_{\diamond + \ddot{u}, \diamond + \bar{u}} {\mathrm{N}_0} {|\bar{\beta}|}^2 E_{\rm d} \mathbb{E} \Big[ \sum_{\bar{k} \in \mathcal{K}} \Omega_{k+\bar{u}-\bar{k}} \Omega_{k+\ddot{u}-\bar{k}}^{*} \Big] & \nonumber \\
& \quad \stackrel{(6)}{=} {\rm Tr}\{ \mathbf{R}^{-1}_{p} \} {\mathrm{N}_0} {|\bar{\beta}|}^2 E_{\rm d} & \\
& \mathbb{E}\{{|\hat{Z}_k^{(3)}|}^2 \} \stackrel{(7)}{=} \sum_{u \in \mathcal{U}} {|\beta_0|}^2 \mathbb{E} \big\{{|(1+\chi^{*}) \Omega_u + \Phi_u|}^2 \big\} \mathrm{N}_0 & \nonumber \\
& \quad \stackrel{(8)}{\leq} {|\beta_0|}^2 \Big[ \big(1 + \chi + \chi^{*} \big) \mu(0,U) & \nonumber \\
& \qquad \quad + \Upsilon_p \big(\mu(0, 2\hat{G}+U) -\mu(0,U)\big) \Big]\mathrm{N}_0  & \\
& \mathbb{E}\{{|\hat{Z}_k^{(4)}|}^2 \} \stackrel{(9)}{=} \sum_{u \in \mathcal{U}} {[\mathbf{R}^{-1}_p]}_{\diamond +u, \diamond +u} {\mathrm{N}_0}^2 & \nonumber \\
& \quad = {\rm Tr}\{ \mathbf{R}^{-1}_{p} \} {\mathrm{N}_0}^2, & \!\!\!\!\!\!\!\!
\end{flalign}
\end{subequations}
where ${\scriptstyle\stackrel{(1)}{\leq}}$ follows from the Cauchy Schwartz inequality; ${\scriptstyle\stackrel{(2)}{=}}$ follows by using the second moment of $\hat{\mathbf{W}}^{(U)}$ in \eqref{eqn_PhN_est_LS}; 
${\scriptstyle\stackrel{(3)}{=}}$ follows by changing the order of summation of $v, \dot{u}, \ddot{u}$; ${\scriptstyle\stackrel{(4)}{=}}$ follows from the definition of $\mathbf{R}_{p}$; ${\scriptstyle\stackrel{(5)}{\leq}}$ by using $\bar{\beta} \triangleq \max_{k \in \mathcal{K}} |\beta_k|$ and increasing the summation range for $\bar{k}$; ${\scriptstyle\stackrel{(6)}{=}}$ follows from \eqref{eqn_lemma_1}; ${\scriptstyle\stackrel{(7)}{=}}$ follows from definition of $\mu(\cdot)$ in \eqref{eqn_Sk_stats}; ${\scriptstyle\stackrel{(8)}{=}}$ follows from \eqref{eqn_revLemma_2}, \eqref{eqn_revLemma_3} and ${\scriptstyle\stackrel{(9)}{=}}$ follows by observing that $\hat{W}_{u}$ and $W_{k+u}$ are independent for $k \in \mathcal{K} \setminus \mathcal{G}$. 

\section{Performance Analysis} \label{sec_perf_anal}
From the analysis in the previous section, the effective channel between the $k$-th OFDM input and the $k$-th demodulated output (after PhN compensation) can be expressed as:
\begin{eqnarray} \label{eqn_equiv_channel}
\hat{Y}_k = \beta^{*}_{0} \beta_{k} x_{k} (1+\chi^{*}) \mu(0,U) + \hat{I}_k + \hat{Z}_k,
\end{eqnarray}
where the statistics of the $\hat{I}_k, \hat{Z}_k$ are discussed in Section \ref{sec_anal_output}. Note that for demodulating $x_k$ from \eqref{eqn_equiv_channel}, the RX requires estimates of the channel coefficients $\{\beta^{*}_{0} \beta_{k} (1+\chi^{*}) \mu(0,U) | k \in \mathcal{K} \}$. These coefficients can be tracked accurately at the RX using pilot symbols, at the boosed SINR after PhN compensation. For brevity, we shall assume perfect estimates of these coefficients for the analysis.\footnote{The impact of such channel estimation errors has been well explored in literature \cite{Goldsmith2004} and can be incorporated by increasing power of $\hat{Z}_k$ appropriately.} 
Noting that $\hat{S}_k, \hat{I}_k, \hat{Z}_k$ in \eqref{eqn_equiv_channel} are uncorrelated, a lower bound to SINR on sub-carrier $k$ can be obtained as:
\begin{align}
\Gamma^{\rm LB}_k(\boldsymbol{\beta}) &= \frac{{|\beta_{0} \beta_{k}|}^2 E_{\rm d} {|1+\chi|}^2 {\mu(0,U)}^2}{ \sum_{i=1}^{2} \bar{\sigma}_{\hat{I}^{(i)}_k}^2 + \sum_{j=1}^{4} \bar{\sigma}_{\hat{Z}^{(j)}_k}^2 }, \label{eqn_SINR_LB}
\end{align}
where $\bar{\sigma}_{\hat{I}^{(i)}_k}^2, \bar{\sigma}_{\hat{Z}^{(j)}_k}^2$ are the upper bounds in \eqref{eqn_Ik_stats} and \eqref{eqn_Z_k_stats}, respectively and $\boldsymbol{\beta} \triangleq \{\beta_k | k \in \mathcal{K}\}$. Furthermore, considering independent demodulation of each sub-carrier, the achievable system throughput (conditioned on $\boldsymbol{\beta}$) can be lower bounded as:\footnote{Here the expression for ergodic capacity is used, by assuming $\boldsymbol{\beta}$ remains constant for infinite time but the PhN $\boldsymbol{\Omega}$ experiences many independent realizations. This capacity is representative of the throughput of practical channel codes that have a block length spanning multiple OFDM symbols but smaller than the coherence time of $\beta_{k}$ \cite{Foschini1998}.}
\begin{align}
C(\boldsymbol{\beta}) & \stackrel{(1)}{=} \frac{1}{K}\sum_{k \in \mathcal{K} \setminus \mathcal{G}} \bigg[ \mathscr{H}(x_k) - \mathscr{H}(x_k | \hat{Y}_k, \widehat{\beta_{0} \boldsymbol{\Omega}^{(U)}} ) \bigg] \nonumber \\
& \geq \frac{1}{K}\sum_{k \in \mathcal{K} \setminus \mathcal{G}} \bigg[ \mathscr{H}(x_k) - \mathscr{H}(x_k | \hat{Y}_k ) \bigg] \nonumber \\
& \stackrel{(2)}{\geq} \frac{1}{K}\sum_{k \in \mathcal{K} \setminus \mathcal{G}} \bigg[ \log\left[ E_{\rm d}\right] - \log\left[ E_{\rm d} - \frac{{\big|\mathbb{E}\{\hat{Y}_k x_k^{*}\} \big|}^2 }{ \mathbb{E}\{{|\hat{Y}_k|}^2\}}\right] \bigg] \nonumber \\
& \stackrel{(3)}{=} \frac{1}{K}\sum_{k \in \mathcal{K} \setminus \mathcal{G}} \log \big(1 + \Gamma_{k}^{\rm LB}(\boldsymbol{\beta}) \big) \triangleq C^{\rm LB}(\boldsymbol{\beta}), \label{eqn_cap}
\end{align}
where ${\scriptstyle\stackrel{(1)}{=}}$ follows by defining $\mathscr{H}(\cdot)$ as the differential entropy; ${\scriptstyle\stackrel{(2)}{\geq}}$ is obtained by using Gaussian signaling for $x_k$ and computing error variance of estimator $\hat{x}_k = \mathbb{E}\{\hat{Y}_k x_k^{*}\} \hat{Y}_k \big/ \mathbb{E}\{{|\hat{Y}_k|}^2\}$; and ${\scriptstyle\stackrel{(3)}{=}}$ is obtained by noting that $\hat{S}_k, \hat{I}_k, \hat{Z}_k$ in \eqref{eqn_equiv_channel} are uncorrelated. Note that the summation limits in \eqref{eqn_cap} ensure that the RS overhead is also considered in the throughput analysis, and the dependence on the TX/RX beamformers is captured in $\boldsymbol{\beta}$.

\section{Near-optimal RS design and system parameters} \label{sec_RS_design}
Although $C^{\rm LB}(\boldsymbol{\beta})$ in \eqref{eqn_cap} is a closed-form throughput lower bound, its dependence on the sub-carrier index $k$ makes it difficult to find near-optimal RS parameters. Therefore, for optimizing parameters we shall focus only on the \emph{typical} data-subcarriers for which $|k| \gg G$. Note that for these sub-carriers, using $G \geq \hat{G}+2U$, criterion C2 and assuming $E_{\rm r} \ll E_{\rm s}$, we observe that $\bar{\sigma}_{\hat{I}^{(1)}_k}^2, \bar{\sigma}_{\hat{Z}^{(1)}_k}^2$ are negligibly small.\footnote{Simulations have shown this approximation to be accurate even for $E_{\rm r}$ as high as $E_{\rm r} = E_{\rm s}/2$ and typical PhN levels.} So ignoring these two terms, \eqref{eqn_SINR_LB} and \eqref{eqn_cap} can be approximated as:
\begin{subequations} \label{eqn_aprx_SINR_cap}
\begin{align}
\Gamma^{\rm aprx}_k(\boldsymbol{\beta}) &= \frac{{|\beta_{0} \beta_{k}|}^2 E_{\rm r} (E_{\rm s} - E_{\rm r}) {|1+\chi|}^2 {\mu(0,U)}^2}{ \Xi(\boldsymbol{\beta})} \label{eqn_SINR_aprx} \\
C^{\rm aprx}(\boldsymbol{\beta}) &= \frac{1}{K}\sum_{k \in \mathcal{K} \setminus \mathcal{G}} \log \big(1 + \Gamma_{k}^{\rm aprx}(\boldsymbol{\beta}) \big) \label{eqn_cap_LB2}
\end{align}
where:
\begin{align}
\Xi(\boldsymbol{\beta}) \triangleq & {|\beta_0|}^2 E_{\rm r} \Big[ {|\bar{\beta}|}^2 (E_{\rm s} - E_{\rm r}) + (K-|\mathcal{G}|) \mathrm{N}_0 \Big] \Big[ \mu(0,U) \nonumber \\
& \times (1+\chi+\chi^{*}) + {\Upsilon}_p \big[ \mu(0, 2\hat{G}+U) - \mu(0,U)\big] \Big] \nonumber \\
& - {|\beta_0 \bar{\beta}|}^2 {|1+\chi|}^2 {\mu(0,U)}^2 (E_{\rm s}-E_{\rm r}) E_{\rm r} \nonumber \\
& + {\rm Tr}\{ E_{\rm r} \mathbf{R}^{-1}_{p} \} {\mathrm{N}_0} \big[ {|\bar{\beta}|}^2 (E_{\rm s}-E_{\rm r}) + (K-|\mathcal{G}|)\mathrm{N}_0 \big] . \label{eqn_xi_def}
\end{align}
\end{subequations}
In this section, we shall analyze the impact of the 5 RS parameters: $G, \hat{G}, \{p_{-\hat{G}},.., p_{\hat{G}}\}, E_{\rm r}$ and $U$ on $C^{\rm aprx}(\boldsymbol{\beta})$, and find the $C^{\rm aprx}(\boldsymbol{\beta})$-maximizing designs of each of them, while the remaining 4 parameters are held constant. The 5 parameters can then be jointly optimized by simply iterating through these conditionally optimal solutions. In practice, however, several of these parameters may be pre-determined constants, thereby obviating the need to iterate through them. 
\subsubsection{Worst case $\chi$}
From \eqref{eqn_chi_defn}, note that $\chi$ is a statistical parameter whose value may be unknown. So here we conservatively use the worst case value of $\chi$ that maximizes $\Xi(\boldsymbol{\beta}) / {|1+\chi|}^2$ (and thus minimizes $C^{\rm aprx}(\boldsymbol{\beta})$) while satisfying constraint \eqref{eqn_revLemma_4}. Note that from \eqref{eqn_xi_def} this term can be expressed in the form:
$$ \Xi(\boldsymbol{\beta}) \big/ {|1+\chi|}^2 = a + (b - c {|\chi|}^2) \big/ {|1+\chi|}^2,$$
where $a,b,c$ are constants independent of $\chi$ and $b- c {|\chi|}^2 > 0$ from \eqref{eqn_revLemma_4}. It can readily be verified that for a given $|\chi|$, this term is maximized when $\chi$ is real and negative. Since the worst-case solution is either a boundary point, singular point or a stationary point of $\Xi(\boldsymbol{\beta}) \big/ {|1+\chi|}^2$, the worst-case choice is the one among the following two solutions which yields a larger value of $\Xi(\boldsymbol{\beta}) \big/ {|1+\chi|}^2$:
\begin{subequations}
\begin{align}
\chi &= \max \Bigg\{- \sqrt{\frac{{\Upsilon}_p [\mu(0, 2\hat{G}+U) - \mu(0,U)]}{\mu(0, U)} }, -1 \Bigg\}, \label{eqn_worst_chi_1}\\
\chi &= \max \Bigg\{- \sqrt{\frac{{\Upsilon}_p [\mu(0, 2\hat{G}+U) - \mu(0,U)]}{\mu(0, U)} }, \nonumber \\
& - \frac{ {\rm Tr}\{ \mathbf{R}^{-1}_{p} \} {\mathrm{N}_0}}{{|\beta_0|}^2 \mu(0,U) } - \frac{ {\Upsilon}_p \big[ \mu(0, 2\hat{G}+U) - \mu(0,U)\big] }{ \mu(0,U) } \Bigg\}. \label{eqn_worst_chi}
\end{align}
\end{subequations}
Here \eqref{eqn_worst_chi_1} results from \eqref{eqn_revLemma_4} and noting that $\Xi(\boldsymbol{\beta}) / {|1+\chi|}^2$ is not analytic at $\chi = -1$ and \eqref{eqn_worst_chi} is from \eqref{eqn_revLemma_4} and the stationary point with $\frac{\partial [ \Xi(\boldsymbol{\beta}) / {|1+\chi|}^2 ]}{ \partial \chi} = 0$. 
\subsubsection{Optimizing $G$}
Using $|\mathcal{G}|=2G+1$, note that $C^{\rm aprx}(\boldsymbol{\beta})$ can be expressed in the form: 
$$C^{\rm aprx}(\boldsymbol{\beta}) = \sum_{k \in \mathcal{K}, |k| > G} \log \Big[1+\frac{\beta_k}{a + b (K-2G-1)} \Big],$$ 
where $a$ and $b$ are positive terms independent of $G$. 
Except in pathological cases where $\beta_{k} \ll \bar{\beta}$ for $k \in \mathcal{G}$ (a deep fade), it can be verified that $C^{\rm aprx}(\boldsymbol{\beta})$ is a strictly decreasing function of $G$. Thus, the throughput optimal value of $G$ is $G^{\rm opt} = \hat{G}+2U$ viz., its lowest allowed value from Section \ref{subsec_PhN_est}. In other words, a lower $G$ (guard band size) is preferred to improve the spectral efficiency. However a minimum size is required to limit the ICI during the PhN estimation process. 
\subsubsection{Optimizing the RS sequence}
We then have the following theorem:
\begin{theorem} \label{Th_opt_RS_design}
For the worst case $\chi$, $C^{\rm aprx}(\boldsymbol{\beta})$ is maximized by an RS with $\mathscr{R}_{p}(u) = 0$ for $u \in \{-2U,...,2U\}\setminus\{0\}$, where $\mathscr{R}_{p}(a) \triangleq \sum_{g = -\hat{G}}^{\hat{G}} p_g p^{*}_{g+a}$, is the aperiodic auto-correlation function of the RS sequence. 
\end{theorem}
\begin{proof}
Note that for a given $\hat{G}$, $\Xi(\boldsymbol{\beta})$ (and hence $C^{\rm aprx}(\boldsymbol{\beta})$) depends on the RS sequence $\{p_{g}| g \in \hat{\mathcal{G}}\}$ via the terms: ${\Upsilon}_p$ and ${\rm Tr}\{\mathbf{R}^{-1}_{p}\}$. From the definition of ${\Upsilon}_p$ in Lemma \ref{Th_PhN_est_interf}, we have:
\begin{align}
0 \leq \Upsilon_p & \stackrel{(1)}{\leq} \lambda_1^{\downarrow}\{\mathbf{R}_{p}^{-1}\}^2 \lambda_1^{\downarrow} \big\{ {[\mathbf{Q}^{(U)}]}^{\dag} \mathbf{P}^{(U)} {[\mathbf{P}^{(U)}]}^{\dag} \mathbf{Q}^{(U)} \big\} \nonumber \\
& \stackrel{(2)}{\leq} \sum_{i = 1}^{2U+1} \sum_{j=1}^{4\hat{G}} \lambda_1^{\downarrow}\{\mathbf{R}_{p}^{-1}\}^2 { \Big|{\big[ {\mathbf{P}^{(U)}}^{\dag} \mathbf{Q}^{(U)} \big]}_{i,j} \Big|}^2 \nonumber \\
& \stackrel{(3)}{\leq} \lambda^{\downarrow}_1\{\mathbf{R}^{-1}_{p}\}^2 4 \hat{G} (2U+1) \max_{a \neq 0} {|\mathscr{R}(a)|}^2 \label{eqn_upsilon_bound} 
\end{align}
where ${\scriptstyle\stackrel{(1)}{\leq}}$ follows from the results on eigenvalue majorization \cite[Eqn 3.20]{Bhatia}, ${\scriptstyle\stackrel{(2)}{\leq}}$ follows from the Frobenius norm bound on the matrix spectral norm and $\stackrel{(3)}{\leq}$ follows by observing that:
\begin{eqnarray}
{\big[ {\mathbf{P}^{(U)}}^{\dag} \mathbf{Q}^{(U)} \big]}_{i,j} = \left\{ \begin{array}{cc}
\mathscr{R}_{p}(-2\hat{G}-i+j) & \text{if $j \leq 2 \hat{G}$} \\
\mathscr{R}_{p}(2U\! + \! 1 \! - \!2\hat{G}-i+j) & \text{if $j > 2 \hat{G}$}. \\
\end{array} \right. \nonumber
\end{eqnarray}
On the other hand, from the definition of $\mathbf{R}_{p}$ in \eqref{eqn_PhN_est_LS} and $\mathbf{P}^{(U)}$ (see also Fig.~\ref{Fig_PhN_est_mtrx_form}), it follows that:
\begin{eqnarray}
\sum_{{i} = 1}^{2U+1} \lambda^{\downarrow}_i \{\mathbf{R}_{p}\} &=& {\rm Tr}\{\mathbf{R}_{p}\} =  (2U+1) E_{\rm r} \nonumber 
\end{eqnarray}
where $\lambda^{\downarrow}_i \{\mathbf{R}_{p}\}$ is the ${i}$-th largest eigenvalue of $\mathbf{R}_{p}$.
Consequently, we obtain the bounds:
\begin{subequations} \label{eqn_RS_theorem_1}
\begin{align}
\lambda^{\downarrow}_1\{\mathbf{R}^{-1}_{p}\} &= \frac{1}{\lambda^{\downarrow}_{2U+1}\{\mathbf{R}_{p}\}} \geq \frac{1}{E_{\rm r}} \\
{\rm Tr}\{\mathbf{R}^{-1}_{p}\} &= |\mathcal{U}| \sum_{i=1}^{2U+1} \frac{1}{|\mathcal{U}|} {[\lambda^{\downarrow}_{i}\{\mathbf{R}_{p}\}]}^{-1} \stackrel{(4)}{\geq} \frac{{2U+1}}{E_{\rm r}}, 
\end{align}
\end{subequations}
where ${\scriptstyle\stackrel{(4)}{\geq}}$ follows from convexity of $f(x) = 1/x$ for $x > 0$ and using Jensen's inequality. These bounds in \eqref{eqn_RS_theorem_1} are satisfied with equality iff $\lambda^{\downarrow}_i \{\mathbf{R}_{p}\} = E_{\rm r} \ \forall {i} \leq 2U+1$. Since $\mathbf{R}_{p}$ is also Hermitian symmetric, this implies \eqref{eqn_RS_theorem_1} is met with equality iff $\mathbf{R}_{p} = E_{\rm r} \mathbb{I}_{2U+1}$. Using \eqref{eqn_upsilon_bound}, \eqref{eqn_RS_theorem_1} and observing that ${[\mathbf{R}_{p}]}_{a,b} = \mathscr{R}(a-b)$, we conclude that both ${\Upsilon}_p$ and ${\rm Tr}\{\mathbf{R}^{-1}_{p}\}$ attain their minimum values when $\mathscr{R}_{p}(u) = 0$ for $u \in \{-2U,...,2U\}\setminus\{0\}$. Finally, noting that $\Xi(\boldsymbol{\beta})$ is an increasing function of ${\Upsilon}_p$ and ${\rm Tr}\{\mathbf{R}^{-1}_{p}\}$, and a larger ${\Upsilon}_p$ also increases the feasible values for the worst-case $\chi$ (see \eqref{eqn_revLemma_4}), the result follows.
\end{proof}
From Theorem \ref{Th_opt_RS_design}, it follows that an optimal choice for the RS is a perfect aperiodic auto-correlation sequence, satisfying $\mathscr{R}_{p}(u) = 0, \ \forall u \neq 0$. A low auto-correlation essentially reduces the power of the PhN estimation noise $\hat{\mathbf{W}}^{(U)}$ and PhN estimation interference $\boldsymbol{\Phi}^{(U)}$ in \eqref{eqn_PhN_est_LS}. While such a perfect auto-correlation is not possible when $\hat{G} > 1$, Barker sequences can provide a close approximation \cite{Golomb1965, Leukhin2018}. Several algorithms have also been proposed in literature to find sequences with near-perfect auto-correlation \cite{Stoica2009}. 
\begin{remark}
With a good choice of $U$, the effective channel \eqref{eqn_equiv_channel} may operate in a noise limited regime where $\sum_{j=1}^{4} \bar{\sigma}_{\hat{Z}^{(j)}_k}^2 > \sum_{i=1}^{2} \bar{\sigma}_{\hat{I}^{(i)}_k}^2$. In such scenarios, it may be sufficient to only minimize ${\rm Tr}\{\mathbf{R}^{-1}_{p}\}$ by having $\mathscr{R}_{p}(u) = 0$ for $u \in \{-2U,...,2U\}\setminus\{0\}$. Thus, aperiodic ZCZ sequences \cite{Donelan2002, Han2007, Wu2005} can also be good candidates for the RS.
\end{remark}
%
\subsubsection{Optimizing $\hat{G}$}
Using $G = \hat{G}+2U$ and assuming the use of a Theorem \ref{Th_opt_RS_design}-satisfying RS for each $\hat{G}$, we also observe that $C^{\rm aprx}(\boldsymbol{\beta})$ is a decreasing function of $\hat{G}$ for a fixed $E_{\rm r}, \chi$. Thus for a given RS energy per symbol $E_{\rm r}$, the throughput-optimal choice of $\hat{G}$ is: 
$$\hat{G}^{\rm opt} = \big\lceil (E_{\rm r}-\bar{E})\big/ {2\bar{E}} \big\rceil, $$ where $\bar{E}$ is the spectral mask limit on energy per sub-carrier (see Section \ref{sec_chan_model}). In other words, a smaller $\hat{G}$ (number of RS sub-carriers) is preferred to improve the spectral efficiency, but a minimum number is also required to satisfy the spectral mask regulations for a given RS power $E_{\rm r}$. As a side note, it should also be ensured that $\hat{G}^{\rm opt}/T_{\rm s}$ is sufficiently smaller than the channel coherence bandwidth for \eqref{eqn_Yk} to hold. 
\subsubsection{Optimizing $E_{\rm r}$}
Note that $E_{\rm r} = 0$ and $E_{\rm r} = E_{\rm s}$ are both poor allocations with $C^{\rm aprx}(\boldsymbol{\beta})=0$, and thus the best power allocation lies somewhere in between. Since $C^{\rm aprx}(\boldsymbol{\beta})$ is differentiable in $E_{\rm r}$, a valid stationary point for $E_{\rm r}$ can be computed by setting the partial derivative $\frac{\partial [ E_{\rm r} (E_{\rm s}-E_{\rm r}) / \Xi(\boldsymbol{\beta} ]}{ \partial E_{\rm r}} = 0$ as:
\begin{align}
E_{\rm r}^{\rm statn} &= \frac{-A_0 + \sqrt{A_0^2 + A_0 E_{\rm s} (A_1+A_2 E_{\rm s})}}{A_1 + A_2 E_{\rm s}}, \label{eqn_Er_opt}
\end{align}
where $A_i$ is the coefficient of ${(E_{\rm r})}^i$ in the cumbersome \eqref{eqn_xi_def}, and we note that ${\Upsilon}_p$ and ${\rm Tr}\{ E_{\rm r} \mathbf{R}^{-1}_{p} \}$ are independent of $E_{\rm r}$. The best choice $E_{\rm r}^{\rm opt}$ is then simply the solution among $\{0, E_{\rm s}, E_{\rm r}^{\rm statn}\}$ with the highest value of $C^{\rm aprx}(\boldsymbol{\beta})$.
Note that if $\hat{G}$ is prefixed and cannot be updated based on \eqref{eqn_Er_opt} to meet the spectral mask bound, then the reference power has to be reduced to $E_{\rm r}^{\rm opt} = \min\{E_{\rm r}^{\rm opt}, |\hat{\mathcal{G}}| \bar{E}\}$.
\subsubsection{Optimizing $U$}
Finally, from \eqref{eqn_xi_def}, it can be verified that the power of the interference term decreases with $U$, while that of the noise term increases with $U$. Similarly, from the summation limits in \eqref{eqn_cap}-\eqref{eqn_cap_LB2} and $G \geq \hat{G}+2U$, it can be observed that the RS overhead increases with $U$. Thus $U$ provides a trade-off between ICI suppression, noise enhancement and RS overhead. While finding a closed form expression for the $C^{\rm aprx}(\boldsymbol{\beta})$-maximizing $U$ (for given $E_{\rm r}, \hat{G}$) is intractable, it can be computed numerically by performing a simple line search of $C^{\rm aprx}(\boldsymbol{\beta})$ over the range $0 \leq U \leq [\min\{K_1,K_2\}-\hat{G}]/2$. 

\section{Simulation Results} \label{sec_sim_results}
For the simulations, we consider a single cell THz system with a $\lambda/2$-spaced $32 \times 8$ ($M_{\rm tx} = 256$) antenna TX and a representative RX with a $\lambda/2$-spaced $8 \times 4$ (${M}_{\rm rx} = 32$) antenna panel. We also assume perfect timing synchronization between the TX and RX. Using prior knowledge of the channel spatial correlation matrix, the TX and RX are assumed to beamform along the strongest cluster of channel MPCs \cite{Alkhateeb2015, Caire2017}. 
The TX transmits one spatial OFDM data stream with $T_{\rm s}=1 \mu$s, $K_1=K_2+1=512$ and transmits the length $11$ Barker sequence for the RS ($\hat{G}=5$). 
Unless otherwise stated, the TX and RX are both assumed to have $f_{\rm c} = 90$GHz oscillators with no carrier frequency offset and the PhN is modeled as a Wiener processes, having variances of $\sigma_{\theta,{\rm tx}}^2=\sigma_{\theta,{\rm rx}}^2=0.1/T_{\rm s}$ (see Appendix \ref{appdix_wiener}). This corresponds to the practically relevant PhN level of $-103$ dBc/Hz at $10$MHz offset \cite{Zhang2015}. The RX is also assumed to have perfect knowledge of $\boldsymbol{\beta}$ and $\mathrm{N}_0$. We first validate the derived analytical results in Section \ref{subsec_validate_results} and shall compare performance of the proposed optimized scheme to other PhN estimation schemes in Section \ref{subsec_compare_schemes}. The results in this section are presented as a function of the average post-beamforming SNR without PhN: $E_{\rm s} {|\beta_{\rm rms}|}^2/{\rm N}_0 K$, where we define ${|\beta_{\rm rms}|}^2 \triangleq \sum_k {|\beta_k|}^2/K $. 

\subsection{Validating analytical results} \label{subsec_validate_results}
To validate the analytically derived results in the paper, we shall compare them to numerical results obtained via Monte-Carlo simulations. 
For the numerical results, the SINR $\Gamma_k(\boldsymbol{\beta})$ for sub-carrier $k$ is estimated, via Monte-Carlo iterations, as the inverse mean-square distance of the soft decoded data-symbols from the transmitted constellation points. Note that this SINR definition is equivalent to the inverse square of the signal error vector magnitude. Similarly, the ergodic throughput is numerically computed using the numeric SINR values as $\sum_{k \in \mathcal{K} \setminus \mathcal{G}} \log[1 + \Gamma_k(\boldsymbol{\beta})]/K$. For the numerical results, the exact value of $\chi$ is statistically computed, while for the analytical bounds we use the worst case $\chi$ from \eqref{eqn_worst_chi}. For easy reproducibility, we consider a sample post-TX beamforming channel matrix $\mathbf{H}(t) \mathbf{t}$ with $\tilde{L}=3$ MPCs, ${\tau}_{\ell}= [0, 20, 40]$ns, angles of arrival $\psi^{\rm rx}_{\rm azi} = [0, \pi/10, -\pi/10]$, $\psi^{\rm rx}_{\rm ele} = [0.45\pi, \pi/2, 0.4\pi]$ and effective post-beamforming amplitudes $\alpha_{\ell} \mathbf{a}_{\rm tx}(\ell)^{\dag} \mathbf{t} = [\sqrt{0.6}, -\sqrt{0.3}, \sqrt{0.1}]$. The RX beamformer is assumed to be $\mathbf{r} = \tilde{\mathbf{a}}_{\rm rx}(\psi^{\rm rx}_{\rm azi}(1), \psi^{\rm rx}_{\rm ele}(1)) \big/ | \tilde{\mathbf{a}}_{\rm rx}(\psi^{\rm rx}_{\rm azi}(1), \psi^{\rm rx}_{\rm ele}(1))|$ (see \eqref{eqn_array_response_planar}).

\subsubsection{Second order statistics}
For the aforementioned channel model and fixed $E_{\rm r} = 0.01 E_{\rm s}$, the second moments of the signal, interference and noise components $\hat{S}_k, \hat{I}_k, \hat{Z}_k$ are compared to their analytical bounds \eqref{eqn_Sk_stats}-\eqref{eqn_Z_k_stats} in Fig.~\ref{Fig_compare_SIZ_stats}. As observed from Fig.~\ref{Fig_compare_SIZ_stats}, the bounds on $\mathbb{E}\{{|\hat{S}_k|}^2\}, \mathbb{E}\{{|\hat{Z}_k|}^2\}$ are tight $\forall k$, while the bound on $\mathbb{E}\{{|\hat{I}_k|}^2\}$ in \eqref{eqn_Ik_stats} is only tight for $|k| \gg 1$. This is due to the looseness of the Cauchy-Schwartz inequality in \eqref{eqn_I1k_stats}. In addition in Fig.~\ref{Fig_compare_SIZ_stats}, we also plot the approximate bounds on $\mathbb{E}\{{|\hat{I}_k|}^2\}, \mathbb{E}\{{|\hat{Z}_k|}^2\}$, obtained after ignoring the terms $\bar{\sigma}^2_{\hat{I}_k^{(1)}}$ and $\bar{\sigma}^2_{\hat{Z}_k^{(1)}}$, as considered in Section \ref{sec_RS_design}. Results show that these approximations are quite accurate and also help significantly improve the tightness of the bound on $\mathbb{E}\{{|\hat{I}_k|}^2\}$. 
For the same system settings, the simulated and analytical SINR for the different sub-carriers are studied as a function of average post-beamforming SNR in Fig.~\ref{Fig_compare_SINR_vs_SNR}. We observe that while the SINR bound $\Gamma^{\rm LB}_k(\boldsymbol{\beta})$ is tight for high frequency sub-carriers, the approximate SINR bound $\Gamma^{\rm aprx}_k(\boldsymbol{\beta})$ yields a tight bound $\forall k$. Note that these observations justify our decision to optimize system parameters using \eqref{eqn_aprx_SINR_cap} in Section \ref{sec_RS_design}.
\begin{figure}[!htb]
\centering
\vspace{-0.4cm}
\subfloat[Statistics comparison]{\includegraphics[width= 0.4\textwidth]{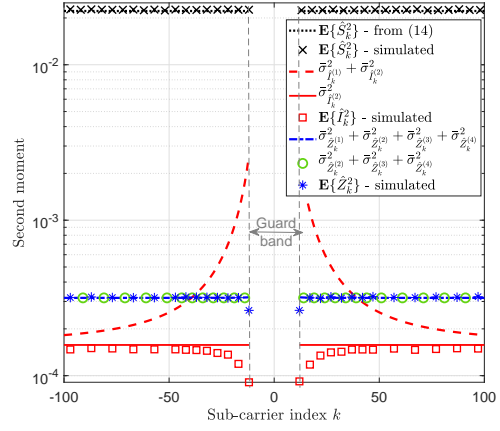} \label{Fig_compare_SIZ_stats}} \hspace{0.5cm}
\subfloat[SINR comparison]{\includegraphics[width= 0.4\textwidth]{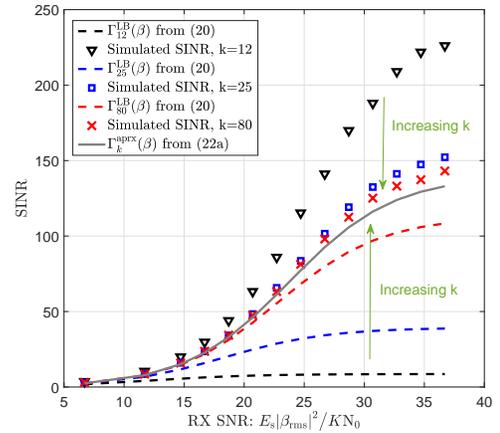} \label{Fig_compare_SINR_vs_SNR}}
\caption{Comparison of the analytical and simulated values of (a) $\mathbb{E}\{{|\hat{S}_k|}^2\}$, $\mathbb{E}\{{|\hat{I}_k|}^2\}$, $\mathbb{E}\{{|\hat{Z}_k|}^2\}$ for $E_{\rm s}{|\beta_{\rm rms}|^2}/({\rm N}_0 K) = 20$ dB and (b) SINR $\Gamma_k(\boldsymbol{\beta})$ for varying SNR. We use $U=3, G=11,\sigma_{\theta}^2=0.2/T_{\rm s}, E_{\rm r} = 0.01E_{\rm s}$ and use 16-QAM modulated data symbols for simulations.}
\label{Fig_verify_SIZ_SINR}
\end{figure}

\subsubsection{Optimal RS parameters}
We next study the impact of the parameters $U, E_{\rm r}$ on the system throughput (including RS overhead) in Fig.~\ref{Fig_Eopt_compare}. Here, while $\mathbf{p}/\|\mathbf{p}\|$ and $\hat{G}$ are held constant, we use the throughput optimal $G, E_{\rm r}$, i.e., $G = \hat{G}+2U$ and $E_{\rm r}$ from  \eqref{eqn_Er_opt}. As observed from the results, the throughput changes unimodally with $U$. While the low throughput at lower $U$ is due to the PhN induced ICI, the poor performance at high $U$ is due to noise accumulation and the RS + null sub-carrier overhead. We also note that the throughput-optimal $U$ increases with SNR. 
We also observe that the analytical bound $C^{\rm LB}(\boldsymbol{\beta})$ is loose for $U \gg 1$, which is due to the looseness of \eqref{eqn_I1k_stats}. However, we observe that $C^{\rm aprx}(\boldsymbol{\beta})$ has an excellent match with the simulated throughput. For simulation results, we additionally also use brute-force search to find the optimal $E_{\rm r}$ value. We observe that the proposed $E_{\rm r}$ value \eqref{eqn_Er_opt} provides almost identical performance to the brute force search result. 
Furthermore, the analytically optimal $U$ coincides with the simulated optimal $U$. For the simulation results, we also include in Fig.~\ref{Fig_Eopt_compare} the performance of the RX when each sub-carrier coefficient of the RS $\mathbf{p}$ is picked randomly from the 16QAM alphabet. As evident from the results, the optimal design of the RS suggested in Section \ref{sec_RS_design} provides significant gains over an arbitrarily picked RS. These observations validate the choice of parameters in Section \ref{sec_RS_design}. 
\begin{figure}[!htb]
\centering
\vspace{-0.4cm}
\includegraphics[width= 0.4\textwidth]{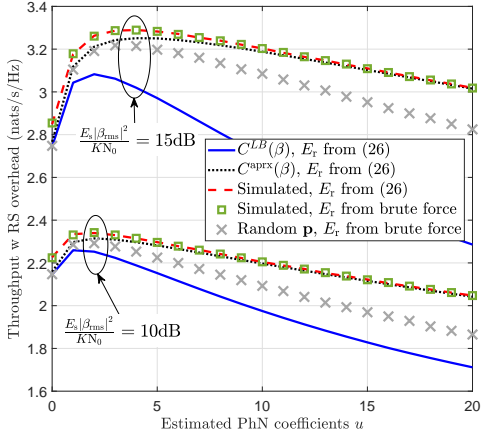}
\vspace{-0.2cm}
\caption{Comparison of the analytical and simulated values of throughput versus $U$ with optimal $E_{\rm r}$. We use $\sigma_{\theta}^2=0.2/T_{\rm s}, \bar{E}=E_{\rm s}$, and we use 16-QAM modulated data symbols for numerically computing capacity.}
\label{Fig_Eopt_compare}
\end{figure}

\subsubsection{LMMSE estimation and frequency offset} Next, we compare the symbol error rate (SER) of the system, averaged over sub-carriers, as a function of SNR in Fig.~\ref{Fig_SER_LS_vs_MMSE} for both LS and LMMSE based PhN estimators for \eqref{eqn_PhN_est_mtrx_form}. Here, while $E_{\rm r}, \hat{G}, G$ are held fixed, we use (i) $U=12$ for the LMMSE estimator and (ii) use $U=U^{\rm opt}$ for the LS estimator, where the $C^{\rm aprx}(\boldsymbol{\beta})$ maximizing $U^{\rm opt}$ is obtained for each SNR value by a line search over $U \in [0, 12]$ as discussed in Section \ref{sec_RS_design}. While not require for the LS estimator, we assume perfect knowledge of the PhN second-order cross-statistics $\Delta_{k_1, k_2}$ for the LMMSE estimator. Results show that knowledge of the cross-statistics $\Delta_{k_1, k_2}$ for $k_1 \neq k_2$ indeed provide some performance gains for the MMSE estimator over the LS estimator. However such knowledge might not be available in practice or can even by faulty. As an example, in Fig.~\ref{Fig_SER_LS_vs_MMSE} we also depict the performance of a scenario where the TX and RX oscillators have a $1$ MHz carrier frequency offset. As is evident, the proposed PhN mitigation technique also provides good performance and outperforms LMMSE estimator if there exists an unknown oscillator frequency mismatch smaller than $U^{\rm opt}/T_{\rm s}$. In other words, it is more resilient to PhN modeling errors than LMMSE.\footnote{While not depicted here for brevity, the simulated SER of LMMSE estimation by ignoring the cross-terms (i.e., using $\Delta_{k_1,k_2} = 0$ for $k_1 \neq k_2$) is almost identical to the LS estimator.}
\begin{figure}[!htb]
\centering
\vspace{-0.4cm}
\includegraphics[width= 0.4\textwidth]{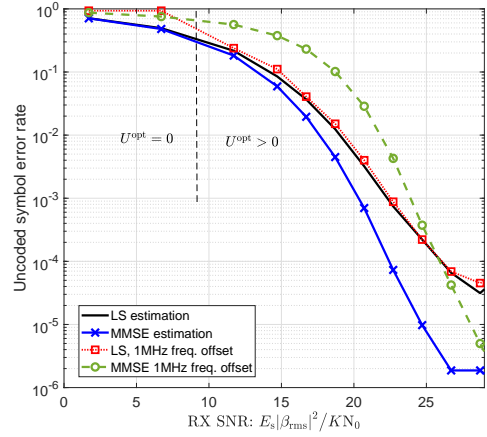}
\vspace{-0.4cm}
\caption{Comparison of SER (averaged over data sub-carriers) versus SNR for LS and LMMSE estimation of PhN. We use $G=30,\sigma_{\theta}^2=0.2/T_{\rm s}, E_{\rm r} = 0.01E_{\rm s}$, $U=12$ for LMMSE, $U \in [0, 12]$ for LS and we use 16QAM modulated data symbols for simulations.}
\label{Fig_SER_LS_vs_MMSE}
\end{figure}

\subsection{Comparison to other schemes} \label{subsec_compare_schemes}
Finally to justify the choice of our scheme, in Fig.~\ref{Fig_SER_compare_schemes} we compare the performance of (i) our optimized PhN estimation technique to  (ii) PhN estimation with sinusoidal RS ($\hat{G}=0$) \cite{Randel2010}, (iii) CPE only estimation \cite{Guo2017}, (iv) Iterative decision feedback based PhN estimation \cite{Wu2002, Petrovic2007} and (v) system with no PhN (an SER lower bound), under the 3GPP Rel. 15  UMa LoS and NLoS channel models \cite{TR38900_chanmodel}. Additionally, we also include the performance of (vi) cyclic prefix aided single carrier transmission with symbol duration $T_{\rm s}/K$ and time domain PhN estimation. While the schemes (i), (ii), (v) do not require channel estimates prior to PhN estimation, for schemes (iii), (iv) and (vi) we assume apriori perfect channel estimates. 
We also assume $U \in [0,20], G$ and $E_{\rm r}$ for schemes (i) and (ii) are designed in accordance to Section \ref{sec_RS_design}, and the same values of $U$ as in scheme (i) are used for scheme (iv). A per sub-carrier spectral mask limit of $\bar{E} = 5 E_{\rm s}/K$ is considered, that limits the pilot/RS power boosting. To keep the total pilot + null sub-carrier overhead comparable across schemes, we use $30$ uniformly-spaced pilot sub-carriers, each with power $\bar{E}$ for schemes (iii) and (iv), while for scheme (ii) we use a single RS with power $\bar{E}$. For the single-carrier transmission scheme (vi), use $30$ uniformly spaced pilots of energy ${E}_{\rm s}/K$ in every $K$ symbols, the RX estimates PhN phase at pilot symbols after zero-forcing channel equalization, and uses piece-wise linear phase interpolation to obtain PhN phase at the intermediate data symbols. For these schemes, the SERs with uncoded 16QAM transmission, averaged over sub-carriers and channel realizations, are presented in Figs.~\ref{Fig_SER_compare_schemes_LoS} and Fig.~\ref{Fig_SER_compare_schemes_nLoS}. Here the RX beamformer is designed as $\mathbf{r} = \tilde{\mathbf{a}}_{\rm rx}(\psi^{\rm rx}_{\rm azi}, \psi^{\rm rx}_{\rm ele}) \big/ | \tilde{\mathbf{a}}_{\rm rx}(\psi^{\rm rx}_{\rm azi}, \psi^{\rm rx}_{\rm ele})|$, where $\tilde{\mathbf{a}}_{\rm rx}(\cdot)$ is from \eqref{eqn_array_response_planar} and $\psi^{\rm rx}_{\rm azi}, \psi^{\rm rx}_{\rm ele}$ are the central azimuth and elevation angles of arrival for the strongest MPC cluster. A similar design is used for the TX beamformer $\mathbf{t}$. 
\begin{figure}[!htb]
\centering
\vspace{-0.4cm}
\subfloat[UMa LoS channel]{\includegraphics[width= 0.4\textwidth]{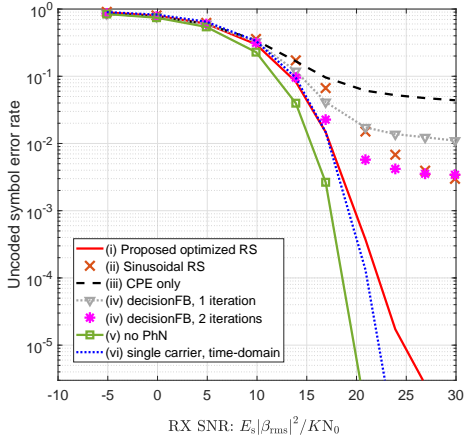} \label{Fig_SER_compare_schemes_LoS}} \hspace{0.5cm}
\subfloat[UMa nLoS channel]{\includegraphics[width= 0.4\textwidth]{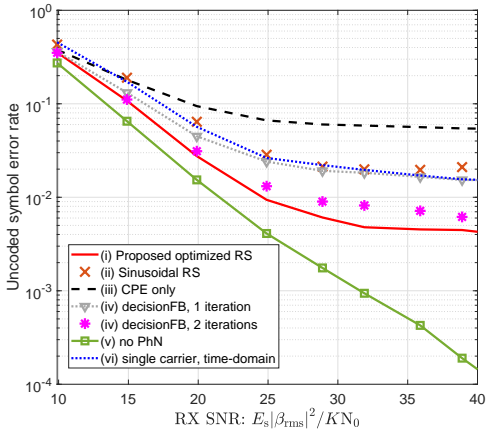} \label{Fig_SER_compare_schemes_nLoS}}
\caption{Comparison of the average SER versus SNR for different PhN compensation schemes (i)--(vi) in the UMa LoS and nLoS channels \cite{TR38900_chanmodel}. Here the SER averaging is performed over data sub-carriers and channel realizations, and we use $\psi^{\rm tx}_{\rm azi}=\psi^{\rm rx}_{\rm azi}=0$ and $\psi^{\rm tx}_{\rm ele}=\psi^{\rm rx}_{\rm ele}=\pi/2$ for the LoS path in Fig.~\ref{Fig_SER_compare_schemes_LoS}.}
\label{Fig_SER_compare_schemes}
\end{figure}

For LoS scenario, we observe that schemes (i) and (vi) significantly outperform (ii), (iii) and (iv), while only requiring a fraction of the computational effort as (iv). Here the poor performance of scheme (ii) is due to the spectral mask limit $\bar{E}$ on the power of the sinusoidal RS and of scheme (iii) is due to the un-compensated ICI after CPE correction. The decision feedback approach scheme (iv) achieves only a limited performance improvement over (iii) even after two decision feedback iterations due to the errors in data decoding from previous iterations, i.e., error propagation. 
For the nLoS scenario, the SER performance of all schemes, including the no PhN bound (v), take a hit due to the higher amount of frequency selective channel fading experienced. Here we observe that the proposed scheme (i) also outperforms scheme (vi), in addition to schemes (ii)-(iv), due to the high ISI between data and pilot symbols experienced by single-carrier systems in nLoS channels (with larger delay spreads). We also observe that for a post-beamforming SNR $< 25$dB, scheme (i) is reasonably close to the no PhN limit (v). Thus, with appropriate PhN mitigation, multi-carrier systems can match, and possibly out-perform, single-carrier systems even in a PhN limited regime. 
Overall, we observe that the proposed optimized PhN compensation technique outperforms several existing solutions, and thus, the presented analysis is a good representation of the state-of-the-art performance achievable in analog beamforming systems by using PhN compensation. 

\section{Conclusion} \label{sec_conclusions}
This paper proposes a novel RS-aided PhN estimation and mitigation technique for multi-antenna multi-carrier high frequency systems, wherein the RS is packed compactly in the frequency domain and is separated from data by null sub-carriers. The frequency concentration of the RS helps decouple PhN estimation from channel estimation and also helps keep the null sub-carrier overhead low, while the null sub-carriers prevent the interference between the RS and data. The detailed mathematical analysis shows that the low frequency spectral components of the PhN dominate its behavior, and they can be estimated using the received sub-carriers in the vicinity of the RS. The analysis also shows that a trade-off exists between the RS overhead, ICI and noise accumulation in the proposed scheme, which can be realized by varying the estimated number of PhN spectral components. Furthermore, we conclude that throughput-optimal designs for the RS include Barker sequences and aperiodic ZCZ sequences. Simulations support the analytical results and show that the derived system parameters can enable achieving near optimal performance. The results also show that the proposed technique can compensate for small frequency offsets between the TX and RX oscillators. Finally, simulations show that the proposed PhN compensation technique outperforms several other existing multi-carrier alternatives and can, in addition, also outperform single-carrier schemes in channels with moderately large delay spreads. 

\begin{appendices}

\section{Wiener PhN model} \label{appdix_wiener}
For the Wiener PhN model at $\star = {\rm tx}/{\rm rx}$, $\theta_{\star}[n]$ is a non-stationary Gaussian process. The increments $\theta_{\star}[n]-\theta_{\star}[n-1]$ are independent, zero-mean Gaussian with a variance of $\sigma^2_{\theta, \star} T_{\rm s}/K$ for each $n$, and $\sigma_{\theta, \star}$ is a parameter quantifying the process. Consequently, the sum PhN $\theta[n] = \theta_{\rm tx}[n]+\theta_{\rm rx}[n]$ is also a Wiener process with parameter $\sigma_{\theta}^2 = \sigma_{\theta, {\rm tx}}^2 + \sigma_{\theta, {\rm rx}}^2$.
\begin{lemma} \label{Lemma_PN_properties}
For the Wiener model, the statistics of the PhN nDFT coefficients satisfy:
\begin{subequations} \label{eqn_PN_lemma}
\begin{flalign}
& \Delta_{k_1,k_2} \triangleq \mathbb{E}\{\Omega_{k_1} \Omega_{k_2}^{*}\} & \nonumber \\
&= \left\{ \begin{array}{cc}
\frac{1}{K^2} \big[\frac{K (1 - {|\rho_1|}^2 )}{( 1-\rho_1^{*} )( 1 - \rho_1) } - \frac{\rho_1^{*} (1 - e^{-\frac{\sigma_{\theta}^2 T_{\rm s}}{2}})}{{(1 - \rho_1^{*})}^2} \!\!\!\!\!\!\!\! & \\
- \frac{\rho_1 (1 - e^{-\frac{\sigma_{\theta}^2 T_{\rm s}}{2}})}{{(1 - \rho_1)}^2}  \big] & \text{for $k_1=k_2$} \!\!\!\!\!\!\!\! \\
\frac{1 - e^{-\frac{\sigma_{\theta}^2 T_{\rm s}}{2}}}{K^2(1 - \rho_2/ \rho_1)} \big[\frac{1}{1 - \rho_2} - \frac{1}{1 - \rho_2^{*}} & \\
+ \frac{1}{1 - \rho_1^{*}} - \frac{1}{1 - \rho_1} \big] & \text{for $k_1 \neq k_2$,} \!\!\!\!\!\!\!\!
\end{array} \right. & \label{eqn_lemma_2}
\end{flalign}
\end{subequations}
for arbitrary integers $k_1,k_2$ and $\sigma_{\theta} > 0$, where $\rho_1 \triangleq e^{- \frac{\sigma_{\theta}^2 T_{\rm s} - {\rm j}4\pi k_1}{2K}}$ and $\rho_2 \triangleq e^{-\frac{\sigma_{\theta}^2 T_{\rm s} - {\rm j}4 \pi k_2}{2K}}$. Furthermore, $\Delta_{k,k}$ is a decreasing function of $|k|$ for $\frac{K}{2 \pi}\cos^{-1}\big( \frac{2 \bar{\rho}}{1 + \bar{\rho}^2} \big) \leq |k| \leq K/2$, where $\bar{\rho} \triangleq e^{- \frac{\sigma_{\theta}^2 T_{\rm s}}{2K}}$. 
\end{lemma}
\begin{proof}
Property \eqref{eqn_lemma_2}, for $k_1 \neq k_2$ can be obtained as follows:
\begin{flalign}
& \Delta_{k_1,k_2} = \mathbb{E}\{\Omega_{k_1} {\Omega_{k_2}}^{*}\} & \nonumber \\
&= \frac{1}{K^2} \sum_{\dot{n}, \ddot{n} = 0}^{K-1} \mathbb{E}\{ e^{-{\rm j}[\theta[\dot{n}]-\theta[\ddot{n}]} \} e^{-{\rm j}2 \pi \frac{[k_1 \dot{n} - k_2 \ddot{n}]}{K}} & \nonumber \\
& \stackrel{(1)}{=} \frac{1}{K^2} \sum_{\dot{n}, \ddot{n} = 0}^{K-1} e^{-\frac{\sigma_{\theta}^2 |\dot{n}-\ddot{n}| T_{\rm s}}{2K}} e^{-{\rm j}2 \pi \frac{[k_1 \dot{n} - k_2 \ddot{n}]}{K}} & \nonumber \\
& \stackrel{(2)}{=} \frac{1}{K^2} \sum_{d = 0}^{K-1} \sum_{\dot{n} = d}^{K-1} e^{-\frac{\sigma_{\theta}^2 d T_{\rm s}}{2K}} e^{-{\rm j}2 \pi \frac{k_2 d}{K}} e^{-{\rm j}2 \pi \frac{(k_1-k_2) \dot{n}}{K}} & \nonumber \\
& \qquad + \frac{1}{K^2} \sum_{d = -K+1}^{0} \sum_{\dot{n} = 0}^{K+d-1} e^{\frac{\sigma_{\theta}^2 d T_{\rm s}}{2K}} e^{-{\rm j}2 \pi \frac{k_2 d}{K}} e^{-{\rm j}2 \pi \frac{(k_1-k_2) \dot{n}}{K}} & \nonumber \\
& \stackrel{(3)}{=} \frac{1}{K^2} \sum_{d = 0}^{K-1} e^{-\frac{\sigma_{\theta}^2 d T_{\rm s}}{2K}} \frac{e^{-{\rm j}2 \pi \frac{k_2 d}{K}} - e^{-{\rm j}2 \pi \frac{k_1 d}{K}} }{ e^{-{\rm j}2 \pi \frac{(k_1-k_2)}{K}} - 1} & \nonumber \\
& \qquad + \frac{1}{K^2} \sum_{d = 0}^{K-1} e^{-\frac{\sigma_{\theta}^2 d T_{\rm s}}{2K}} \frac{ e^{{\rm j}2 \pi \frac{k_1 d}{K}} - e^{{\rm j}2 \pi \frac{k_2 d}{K}}}{ e^{-{\rm j}2 \pi \frac{(k_1-k_2)}{K}} - 1} & \nonumber \\
& \stackrel{(4)}{=} \frac{1 - e^{-\frac{\sigma_{\theta}^2 T_{\rm s}}{2}}}{K^2(e^{-{\rm j}2 \pi \frac{(k_1-k_2)}{K}} - 1)} & \nonumber \\
& \qquad \bigg[\frac{1}{1 - \rho_2^{*}} - \frac{1}{1 - \rho_2} + \frac{1}{1 - \rho_1} - \frac{1}{1 - \rho_1^{*}} \bigg] & \label{eqn_lemma1_proof1}
\end{flalign}
where ${\scriptstyle\stackrel{(1)}{=}}$ follows by using the expression for the characteristic function of the Gaussian random variable $\theta[\dot{n}] - \theta[\ddot{n}]$; ${\scriptstyle\stackrel{(2)}{=}}$ follows by defining $d = \dot{n} - \ddot{n}$ and observing that the second term is zero for $d=0$; ${\scriptstyle\stackrel{(3)}{=}}$ follows by computing the sum of the geometric series over $\dot{n}$, and ${\scriptstyle\stackrel{(4)}{=}}$ follows by computing the sum of the geometric series over $d$ and using the definition of $\rho_1, \rho_2$ in the lemma statement. Similarly, for the case of $k_1=k_2$, we have:
\begin{flalign}
& \Delta_{k_1,k_1} \stackrel{(5)}{=} \frac{1}{K^2} \bigg[-K + \sum_{d = 0}^{K-1} \sum_{\dot{n} = d}^{K-1} e^{-\frac{\sigma_{\theta}^2 d T_{\rm s}}{2K}} e^{-{\rm j}2 \pi \frac{k_1 d}{K}} & \nonumber \\
& \qquad + \sum_{d = 0}^{K-1} \sum_{\dot{n} = 0}^{K-d-1} e^{- \frac{\sigma_{\theta}^2 d T_{\rm s}}{2K}} e^{{\rm j}2 \pi \frac{k_1 d}{K}} \bigg] & \nonumber \\
& = \frac{1}{K^2} \bigg[-K + \sum_{d = 0}^{K-1} (K-d) \Big[ {(\rho_1^{*})}^d + {(\rho_1)}^d \Big] \bigg] & \nonumber \\
& \stackrel{(6)}{=} \frac{1}{K^2} \bigg[-K + \frac{K}{1-\rho_1^{*}} + \frac{K}{1 - \rho_1} & \nonumber \\
& \qquad - \frac{\rho_1^{*} (1 - e^{-\frac{\sigma_{\theta}^2 T_{\rm s}}{2}})}{{(1 - \rho_1^{*} )}^2} - \frac{\rho_1 (1 - e^{-\frac{\sigma_{\theta}^2 T_{\rm s}}{2}})}{{(1 - \rho_1 )}^2}  \bigg] & \nonumber \\
& = \frac{1}{K^2} \bigg[\frac{K \big(1 - e^{- \frac{\sigma_{\theta}^2 T_{\rm s}}{K}} \big)}{( 1-\rho_1^{*} ) ( 1 - \rho_1 ) } & \nonumber \\
& \qquad - \frac{\rho_1^{*} (1 - e^{-\frac{\sigma_{\theta}^2 T_{\rm s}}{2}})}{{(1 - \rho_1^{*})}^2} - \frac{\rho_1 (1 - e^{-\frac{\sigma_{\theta}^2 T_{\rm s}}{2}})}{{(1 - \rho_1 )}^2}  \bigg] & \label{eqn_lemma1_proof2}
\end{flalign}
where ${\scriptstyle\stackrel{(5)}{=}}$ follows by using $k_1=k_2$ and $d = \dot{n} - \ddot{n}$ in ${\scriptstyle\stackrel{(1)}{=}}$ of \eqref{eqn_lemma1_proof1} and ${\scriptstyle\stackrel{(6)}{=}}$ follows by using the expression for the sum of an arithmetico-geometric series. Using \eqref{eqn_lemma1_proof1} and \eqref{eqn_lemma1_proof2}, property \eqref{eqn_lemma_2} follows. Next, for $\bar{\rho} = e^{-\frac{\sigma_{\theta}^2}{2K}}$, let us define the function:
\begin{eqnarray}
\mathcal{F}(x) \triangleq \frac{(1 - {\bar{\rho}}^2 )}{K( 1+{\bar{\rho}}^2 - 2 \bar{\rho} x ) } + \frac{2 \bar{\rho} (1 - {\bar{\rho}}^K)[ 2{\bar{\rho}} - (1+ {\bar{\rho}}^2) x]}{K^2{( 1+{\bar{\rho}}^2 - 2 \bar{\rho} x )}^2 }. \nonumber
\end{eqnarray}
It can then be shown from \eqref{eqn_lemma_2} that $\Delta_{k,k} = \mathcal{F}(\cos(2\pi k/K))$. Furthermore, it can readily be observed that both terms of $\mathcal{F}(x)$ are increasing functions of $x$ for $x \leq 2 \bar{\rho}/(1+\bar{\rho}^2)$. Thus, $\Delta_{k,k}$ is an increasing function of $\cos(2 \pi k/K)$ for $\cos(2 \pi k/K) \leq 2 \bar{\rho}/(1+\bar{\rho}^2)$. Finally since $\cos(2 \pi k/K)$ is a decreasing function of $|k|$ for $0 \leq |k| \leq K/2$, the result follows.
\end{proof}
These second-order statistics are depicted for typical $\sigma_{\theta}^2$ values in Fig.~\ref{Fig_DFT_stats}. From the monotone decreasing property of $\Delta_{k,k}$ with $|k|$ and from Fig.~\ref{Fig_DFT_stats}, we observe that most of the PhN power is concentrated in a few dominant spectral components around $k=0$, thus justifying criterion C2. 
\begin{figure}[!htb]
\centering
\vspace{-0.4cm}
\subfloat[$\Delta_{k_1,k_2} = \mathbb{E}\{\Omega_{k_1} \Omega_{k_2}^{*}\}$]{\includegraphics[width= 0.46\textwidth]{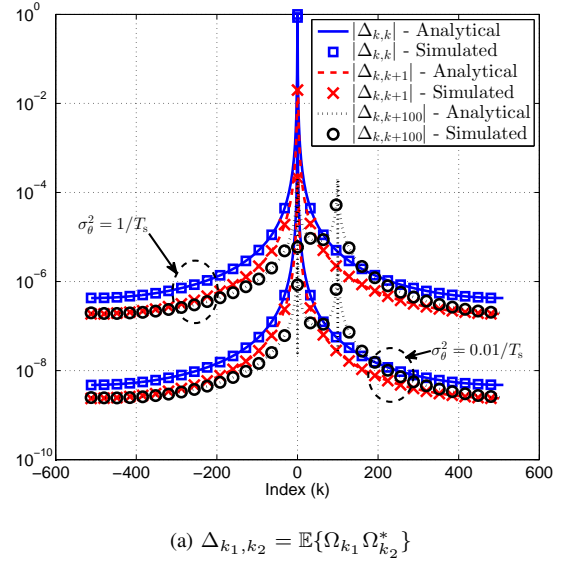} \label{Fig_DFT_stats}} \hspace{0.5cm}
\caption{Comparison of analytical and simulated statistics of the nDFT coefficients of a Wiener PhN process with $T_{\rm s} =1 \mu$s, $K_1=K_2+1=512$.}
\label{Fig_verify_DFT_theorem}
\end{figure}
\end{appendices}

\bibliographystyle{ieeetr}
\bibliography{references}

\begin{IEEEbiography}[{\includegraphics[width=1in,height=1.25in,clip,keepaspectratio]{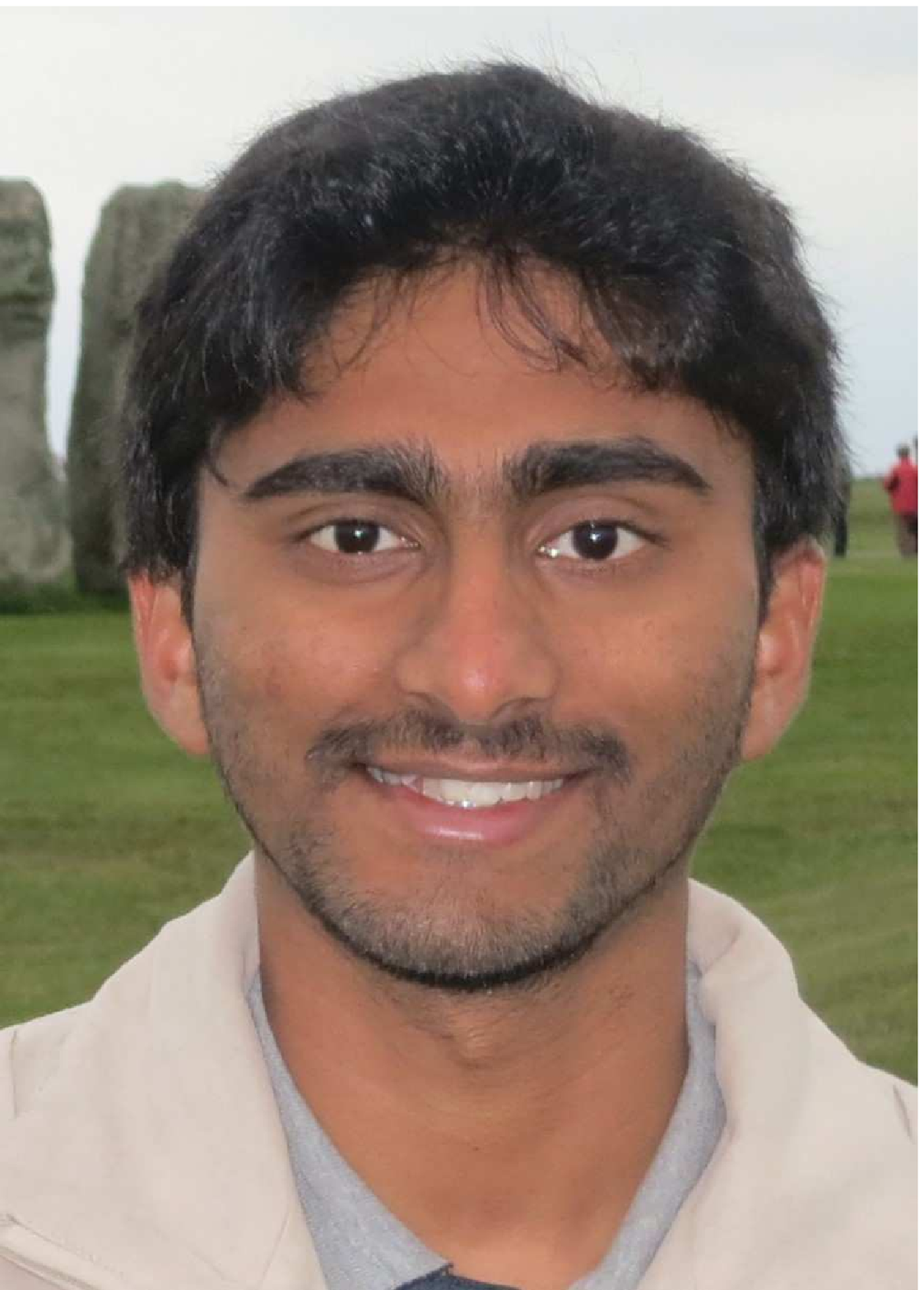}}]{Vishnu V. Ratnam}
(S'10--M'19) received the B.Tech. degree (Hons.) in electronics and electrical communication engineering from IIT Kharagpur, Kharagpur, India in 2012, where he graduated as the Salutatorian for the class of 2012. He received the Ph.D. degree in electrical engineering from University of Southern California, Los Angeles, CA, USA in 2018. He is currently a senior research engineer in the Standards and Mobility Innovation Lab at Samsung Research America, Plano, Texas, USA. His research interests are in AI for wireless, mm-wave and Terahertz communication, massive MIMO, channel estimation and manifold signal processing, and resource allocation problems in multi-antenna networks. 

Dr. Ratnam is a recipient of the Best Student Paper Award at the IEEE International Conference on Ubiquitous Wireless Broadband (ICUWB) in 2016, the Bigyan Sinha memorial award in 2012 and is a member of the Phi-Kappa-Phi honor society.
\end{IEEEbiography}

\end{document}